\documentclass[11pt,reqno,a4paper,onecolumn,accepted=2024-02-27,amssymb,amsmath]{quantumarticle}
\pdfoutput=1

\usepackage[ascii]{inputenc}
\usepackage[bookmarksnumbered,hypertexnames=false,colorlinks=true,linkcolor=blue,urlcolor=blue,citecolor=blue,anchorcolor=green,breaklinks=true,pdfusetitle]{hyperref}
\usepackage{bbm,braket,microtype,mathrsfs,amsmath,amssymb,xcolor,amsthm,mathtools,fullpage,graphicx,enumitem,caption,booktabs,bm,xspace}
\usepackage[longend,ruled,linesnumbered,nokwfunc]{algorithm2e}
\numberwithin{algocf}{section}
\SetKwInput{Input}{Input}
\SetKwInput{Output}{Output}
\SetKwInput{Guarantee}{Guarantee}
\SetKwInput{Analysis}{Analysis}
\SetKw{True}{True}
\SetKw{False}{False}

\usepackage{footnote}
\makesavenoteenv{tabular}
\makesavenoteenv{table}
\usepackage{multirow}
\usepackage[T1]{fontenc}
\usepackage[english]{babel}
\usepackage[capitalize]{cleveref}
\crefalias{AlgoLine}{line}
\crefname{line}{line}{lines}
\usepackage{thm-restate}
\newtheorem{thm}{Theorem}[section]\crefname{thm}{Theorem}{Theorems}
\newtheorem{lem}[thm]{Lemma}\crefname{lem}{Lemma}{Lemmas}
\crefname{clm}{Claim}{Claims}
\newtheorem{dfn}[thm]{Definition}\crefname{dfn}{Definition}{Definitions}
\newtheorem{prp}[thm]{Proposition}\crefname{prp}{Proposition}{Propositions}
\crefname{prb}{Problem}{Problems}
\crefname{rem}{Remark}{Remarks}
\newtheorem{cor}[thm]{Corollary}\crefname{cor}{Corollary}{Corollaries}
\crefname{con}{Conjecture}{Conjectures}
\numberwithin{equation}{section}

\DeclareMathOperator{\poly}{poly}
\DeclareMathOperator{\polylog}{polylog}

\renewcommand{\vec}[1]{\bm{#1}}

\newcommand{\CC}{\mathbb C}

\newcommand{\RR}{\mathbb R}

\newcommand{\Exp}{\mathbb E}
\newcommand{\EE}{\mathbb E}
\newcommand{\U}{\mathrm U}

\newcommand{\Hyp}{\mathrm{Hyp}}
\newcommand{\Bin}{\mathrm{Bin}}

\newcommand{\eps}{\varepsilon}
\renewcommand{\epsilon}{\varepsilon}
\newcommand{\dTV}{d_{\mathrm{TV}}}

\newcommand{\kest}{k_{\mathrm{est}}}
\newcommand{\klb}{k_{\mathrm{lb}}}
\newcommand{\kub}{k_{\mathrm{ub}}}
\newcommand{\vmax}{v_{\mathrm{max}}}
\newcommand{\est}{\mathrm{est}}

\DeclarePairedDelimiter{\parens}{\lparen}{\rparen}
\newcommand{\dynamicparens}[1]{\mathchoice{\parens*{#1}}{\parens{#1}}{\parens{#1}}{\parens{#1}}}
\newcommand{\bigO}[1]{O\dynamicparens{#1}}

\DeclareMathOperator{\Unif}{Unif}

\usepackage{xparse}

\NewDocumentCommand{\bigOt}{o m}{%
  \IfNoValueTF{#1}
    {\ensuremath{\wt{O}\dynamicparens{#2}}}
    {\ensuremath{\wt{O}_{#1}\dynamicparens{#2}}}%
}

\renewcommand{\vec}[1]{\boldsymbol{\mathbf{#1}}}
\DeclarePairedDelimiter{\abs}{\lvert}{\rvert}
\DeclarePairedDelimiter{\norm}{\lVert}{\rVert}
\DeclarePairedDelimiterX{\ip}[2]{\langle}{\rangle}{#1,#2}
\allowdisplaybreaks[4]

\makeatletter
\newcommand*\wt[1]{\mathpalette\wthelper{#1}}
\newcommand*\wthelper[2]{%
        \hbox{\dimen@\accentfontxheight#1%
                \accentfontxheight#11.1\dimen@
                $\m@th#1\widetilde{#2}$%
                \accentfontxheight#1\dimen@
        }%
}

\newcommand*\accentfontxheight[1]{%
        \fontdimen5\ifx#1\displaystyle
                \textfont
        \else\ifx#1\textstyle
                \textfont
        \else\ifx#1\scriptstyle
                \scriptfont
        \else
                \scriptscriptfont
        \fi\fi\fi3
}
\makeatother

\newcommand{\createsubroutineref}[3][]{
  \ifx#1{}%
    \newcommand{#2}{\hyperref[#3]{\texttt{\upshape\ref*{#3}}}}%
  \else%
    \newcommand{#2}{\hyperref[#3]{\texttt{\upshape#1}}}%
  \fi%
}
\createsubroutineref[Grover\-Multiple\-Fast]{\GroverMultipleFast}{alg:grover fast}
\createsubroutineref[Grover\-Certainty]{\GroverCertainty}{alg:grover certainty}
\createsubroutineref[Grover\-Expectation]{\GroverExpectation}{alg:grover expectation}
\createsubroutineref[Grover\-Coupon]{\GroverCoupon}{alg:grover coupon collector}
\createsubroutineref[Grover\-Certainty\-Multiple]{\GroverCertaintyMultiple}{alg:grover multi certainty large gate overhead}
\createsubroutineref[Approx\-Sum]{\ApproxSum}{alg:approx sum}
\createsubroutineref[Approx\-Count]{\ApproxCount}{alg:approx counting mult error}
\createsubroutineref[AmpEst]{\AmpEst}{alg:ampest}
\createsubroutineref[Grover$_{2/3}$]{\GroverLB}{alg:grover prob}

\title{Basic quantum subroutines: finding multiple marked elements and summing numbers}
\date{}
\author{Joran van Apeldoorn}
\affiliation{IViR and QuSoft, University of Amsterdam, The Netherlands}
\email{work@bitofbytes.com}
\author{Sander Gribling}
\affiliation{Department of Econometrics and Operations Research, Tilburg University, Tilburg, The
Netherlands}
\email{s.j.gribling@tilburguniversity.edu}
\orcid{0000-0002-6817-2971}
\author{Harold Nieuwboer}
\affiliation{Korteweg--de Vries Institute for Mathematics and QuSoft, University of Amsterdam, The Netherlands and Faculty of Computer Science, Ruhr University Bochum, Germany and Department of Mathematical Sciences, University of Copenhagen, Denmark}
\email{hani@math.ku.dk}
\orcid{0000-0003-3627-3636}

\begin{document}
\maketitle
\begin{abstract}
  We show how to find all $k$ marked elements in a list of size $N$ using the optimal number~$O(\sqrt{N k})$ of quantum queries and only a polylogarithmic overhead in the gate complexity, in the setting where one has a small quantum memory. Previous algorithms either incurred a factor $k$ overhead in the gate complexity, or had an extra factor $\log(k)$ in the query complexity.

  We then consider the problem of finding a multiplicative $\delta$-approximation of $s = \sum_{i=1}^N v_i$ where $v=(v_i) \in [0,1]^N$, given quantum query access to a binary description of $v$. We give an algorithm that does so, with probability at least $1-\rho$, using $O(\sqrt{N\log(1/\rho)/\delta})$ quantum queries (under mild assumptions on~$\rho$). This quadratically improves the dependence on $1/\delta$ and $\log(1/\rho)$ compared to a straightforward application of amplitude estimation. To obtain the improved $\log(1/\rho)$ dependence we use the first result.
\end{abstract}

\section{Introduction}

\subsection{Finding multiple marked elements in a list} \label{sec:intro grover}
Grover's famous search algorithm~\cite{grover1996QSearch} can be used to find a marked element in a list quadratically faster than possible classically. Formally it can be used to solve the following problem: given a bit string $x \in \{0,1\}^N$, $x \neq 0$, find an index $i \in [N]$ such that $x_i=1$.

In this work we consider the problem of finding \emph{all} indices $i \in [N]$ for which $x_i=1$.
We give a query-optimal quantum algorithm with polylogarithmic gate overhead in the setting where one has a \emph{small quantum memory}.
We explain below why this last assumption makes the problem non-trivial.
This improves over the previous state-of-the-art: previous algorithms were either query-optimal but with a polynomial gate overhead, or had a polylogarithmic gate overhead but also a logarithmic overhead in the query count.

A well-known query-optimal algorithm for the problem is as follows~\cite[Lem.~2]{degraafQuantumVersionsYao2002}.
Let~$k$ be the Hamming weight~$\abs{x} := \sum_{i=1}^N x_i$ of~$x$.
For ease of exposition, suppose the algorithm knows~$k$.
(For our results we will work with weaker assumptions such as knowing only an upper bound on~$k$, or an estimate of it, see \cref{sec:fastgrover}. We also ignore failure probabilities in this part of the introduction.)
A variant of Grover's algorithm~\cite{bbht:bounds} can 
find a single marked element using~$\bigO{\sqrt{N/k}}$ quantum queries and~$\bigO{\sqrt{N/k}\log(N)}$ additional single- and two-qubit gates.
One can then find all~$k$ marked elements using
\[
    \bigO{\sqrt{N/k} + \sqrt{N/(k-1)} + \ldots + \sqrt{N}}
    = \bigO{\sqrt{Nk}}
\]
quantum queries to~$x$. The above complexity is obtained as follows. Suppose we have already found a set~$J \subseteq [N]$ of marked elements. Then to find a new marked element, we replace~$x$ by the string~$z \in \{0,1\}^N$ defined as
\[
    z_i = \begin{cases} x_i & \text{ if } i \not \in J, \\
        0 & \text{ otherwise}.
    \end{cases}
\]
A quantum query to $z$ can be made using a single quantum query to $x$ and quantum query to $J$ (which on input $\ket{i}\ket{b}$ for $i\in [N], b \in \{0,1\}$ returns $\ket{i}\ket{b \oplus \delta_{i \in J}}$ where $\delta_{i \in J} \in \{0,1\}$ is one iff $i \in J$).
In particular, if~$J$ can be stored in a quantum memory (i.e. queried and updated in unit time), then the query complexity will be~$\bigO{\sqrt{Nk}}$ and the time complexity is~$\bigOt{\sqrt{Nk}}$. We refer the interested reader to \cite{Giovannetti:QRAM2008} and \cite[Sec.~5]{Ciliberto:QMachineLearning2018} for a discussion of quantum memory and its (dis)advantages. 

However, when we cannot store $J$ in a quantum memory, a naive implementation of the quantum queries to~$J$ is expensive in terms of gate complexity: if~$\abs{J} = s$, then one can use~$\bigO{s \, \log(N)}$ quantum gates to implement a single query to~$J$.\footnote{We ignore here the cost of maintaining a classical data structure for~$J$, but comment on this again later.} 
Since the size of~$J$ grows to~$k$, the total gate complexity of finding all marked elements will scale as~$\bigOt{\sqrt{N} k^{3/2}}$, which is a factor~$k$ larger than the query complexity.
We show that this factor of~$k$ in the gate complexity can be avoided: we give an algorithm that finds, with large probability, all~$k$ indices using the optimal number of quantum queries to $x$, $\bigO{\sqrt{Nk}}$, while incurring only a polylogarithmic overhead in the gate complexity, in the case where we only have a small quantum memory.
We state a simplified version of our main result below; for the full version, see \cref{thm:all-sol-grover-fast} and the corresponding algorithm \GroverMultipleFast{}.
\begin{thm}
   \label{thm:all-sol-grover-fast-intro}
   Let~$x \in \{0,1\}^N$ with~$\abs{x} = k \geq 2$, and let~$\rho \in (0,1)$ be such that~$k \in \Omega(\log(k/\rho)^3)$ (e.g.~$\rho = \Omega(1/\!\poly(k))$).
   Then we can find, with probability~$\geq 1 - \rho$, all~$k$ indices~$i \in [N]$ for which~$x_i = 1$ using~$\bigO{\sqrt{N k}}$ quantum queries and~$\bigO{\sqrt{Nk} \log(k)^3 \log(N)}$ additional gates.
\end{thm}
We mention that by a simple coupon-collector argument one can already achieve both query- and gate-complexity $\sqrt{Nk}\polylog(N,1/\rho)$, see \cref{prop:most-sol-grover}.
Our algorithm completely removes the $\polylog(N)$ factors in the query complexity and moreover has a much improved dependence on $\log(1/\rho)$: one can achieve $\rho = 1/\!\poly(k)$ without increasing the number of quantum queries made by the algorithm.
In the same spirit, we mention that previous work had already shown that simply boosting a constant success probability is not optimal for finding a single marked element: one can do so with probability $\geq 1-\rho$ using $\sqrt{N \log(1/\rho)}$ quantum queries~\cite{buhrmanBoundsSmallerrorZeroerror1999}.

In a nutshell, our algorithm is a hybrid between the quantum coupon-collector and the query-optimal algorithm described above.
First, we use the coupon collection strategy to find~$t$ marked indices~$1 \leq i_1 < \dots < i_{t} \leq n$, for $t$ roughly $k / \! \log(k)^2$.
A basic property of this strategy is that the resulting indices~$\{i_1, \dotsc, i_{t}\}$ yield a uniformly random subset of size $t$ of the marked indices in $x$. 
Next, for every~$j \in [t+1]$, we use the query-optimal algorithm to find all remaining marked elements in the interval~$(i_{j-1}, i_{j}) \subseteq [n]$, where we write~$i_0 = 0$ and~$i_{t + 1} = n+1$. 
With high probability over the found indices~$\{i_1, \dotsc, i_t\}$, each of the intervals~$(i_{j-1}, i_j)$ contains few remaining marked indices, which reduces the effect of the high gate-complexity overhead of the query-optimal search algorithm.

\subsection{Improved quantum summing algorithm}

Given quantum query access to a binary description of~$v \in [0,1]^N$, how difficult is it to obtain, with probability~$\geq 1 - \rho$, a multiplicative $\delta$-approximation\footnote{Here we use the convention that a multiplicative $\delta$-approximation of a real number~$s$ is a real number~$\tilde s$ for which~$(1-\delta)s \leq \tilde s \leq (1+\delta)s$.} of the sum~$s = \sum_{i=1}^N v_i$? We provide an algorithm to do so whose complexity can be tuned by choosing a parameter~$p \in (0,1)$; one special case of our second main result is as follows, see \cref{thm:fast-summing} for the full version.
In the version below we have made very mild assumptions on the failure probability~$\rho$ and precision~$\delta$, which essentially correspond to the regime in which one makes at most $\bigO{N}$ quantum queries.

\begin{thm}[Informal version of \cref{thm:fast-summing}]
    Let $v \in [0,1]^N$. Let $\rho, \delta \in (0,1)$ be such that $\log(1/\rho)/{\delta} = O(N)$.
    Then we can find, with probability $\geq 1-\rho$, a multiplicative $\delta$-approximation of $\sum_{i=1}^N v_i$ using
    \begin{equation}
      \label{eq:informal summing complexity intro}
      \bigO{\sqrt{\frac{N}{\delta} \log(1/\rho)}}
    \end{equation}
    quantum queries to binary descriptions of the entries of~$v$, and a gate complexity which is larger by a factor polylogarithmic in $N$, $1/\delta$ and $1/\rho$.
\end{thm}
In a nutshell, our algorithm first finds all indices of ``large enough'' entries of the~$v$ using~\GroverMultipleFast{} and sums the corresponding elements classically.
It then rescales the remaining ``small enough'' elements and uses amplitude estimation~\cite{bhmt:countingj} to approximate their sum.
To determine what ``large enough'' means, we use a recent quantum quantile estimation procedure from \cite{hamoudi:QuantumSubGaussianMean2021}.
Choosing the quantile carefully controls both the number of elements that need to be found in the first stage, as well as the size of the elements that remain to be summed in the second stage.
Note that it is the above version of Grover's algorithm that allows us to obtain a query complexity with only a~$\sqrt{\log(1/\rho)}$-dependence, and without additional polylogarithmic factors in~$N$ and~$\delta$.
Indeed, the fact that the number of quantum queries required to find multiple marked elements does not depend on $\log(1/\rho)$ (for $\rho$ not too small) allows us to balance the complexities of the two stages.\footnote{{For completeness we mention that if you \emph{do} allow a large quantum memory, or a polynomial overhead in the gate complexity, then the same $\sqrt{\log(1/\rho)}$-dependence can be obtained using an analogous approach as the one we use to prove \cref{thm:fast-summing}, but instead relying on, e.g., one of the versions of Grover's algorithm discussed in \cref{sec:intro grover}.}}

The problem we consider can be viewed as a special case of the mean estimation problem, or as a generalization of the approximate counting problem for binary strings $x \in \{0,1\}^N$. We briefly discuss how our results compare to prior work on those problems.

\paragraph{Mean estimation algorithms.}
After multiplying the $v_i$ by a factor $\frac{1}{N}$, we can interpret the problem of finding a multiplicative $\delta$-approximation of the sum $s = \sum_{i=1}^N v_i$ as the problem of obtaining a multiplicative $\delta$-approximation of the mean $\mu = \frac{1}{N} \sum_{i=1}^N v_i$ of the random variable that, for each $i \in [N]$, takes value $v_i$ with probability $1/N$.
Quantum algorithms for the mean estimation problem date back to the work of Grover~\cite{grover:telecomputation,grover1998framework}.
A careful application of maximum finding and quantum amplitude estimation yields such an approximation of $\mu$, with probability $\geq 1-\rho$,  using $\bigO{\frac{\sqrt{N}}{\delta}\log(1/\rho)}$ quantum queries and polylogarithmic gate overhead, see \cref{thm:normal multiplicative mean est}.
We improve the dependence on $\delta$ from $1/\delta$ to $1/\sqrt{\delta}$.

As for applications, we note that \cref{thm:normal multiplicative mean est} was used to give quantum speedups for the matrix scaling problem in~\cite{qscalingICALP,qscaling2STACS}, where it is used to approximate the row- and column sums of a matrix with non-negative entries.
This is one of their main sources of quantum speedup, and the quality of this approximation directly affects the achievable precision for the matrix scaling problem.
Using the improved quantum summing subroutine of \cref{thm:fast-summing}, the dependence on the desired precision~$\eps$ for the matrix scaling problem is further improved.
More precisely, if~$A \in \RR_{\geq 0}^{N \times N}$ is an~$N \times N$ matrix with non-negative entries, let~$r(A) \in \RR_{\geq 0}^N$ denote its \emph{vector of row sums}, i.e., $r_i(A) = \sum_{j=1}^N A_{ij}$.
Then given quantum query access to~$A$, using the improved summing subroutine, one can with~$\bigOt{N^{1.5} / \sqrt{\delta}}$ queries compute a vector~$\hat{r} \in \RR^n$ such that \begin{equation*}
  \norm{\hat{r} - r(A)}_1 \leq \delta \norm{r(A)}_1.
\end{equation*}
Computing such an~$\hat r$ with~$\delta = \eps^2$ is the bottleneck in the second-order method for matrix scaling presented in~\cite{qscaling2STACS}.
By reducing the complexity of this step, this method is improved to become better than the fastest \emph{classical} first-order method (Sinkhorn's algorithm) 
for entrywise-positive matrices: the classical method finds an~$\eps$-$\ell^1$-scaling of entrywise positive matrices in time~$\bigOt{N^2 / \eps}$, whereas the box-constrained Newton method now runs in time~$\bigOt{N^{1.5} / \eps}$. Note that this gives an algorithm for matrix scaling whose runtime is sublinear in the input size when $1/\eps = o(\sqrt{N})$, corresponding to $1/\delta = o(N)$, which is precisely the regime of $\delta$ for which the quantum subroutine improves over classical summing. 

We remark that faster mean estimation algorithms have been developed for example for random variables with a small variance $\sigma^2$.
Indeed, the current state of the art obtains a multiplicative $\delta$-approximation, with probability $\geq 1-\rho$, using $\bigOt{\parens{\frac{\sigma}{\delta\mu} + \frac{1}{\sqrt{\delta \mu}}} \log(1/\rho)}$ quantum queries in expectation~\cite{hamoudi:QuantumSubGaussianMean2021,kothari:source}.\footnote{In \cite[Proposition~6.4]{hamoudi:QuantumSubGaussianMean2021}, a matching (up to log-factors) lower bound is shown for Bernoulli random variables. We remark that our algorithm does not break that lower bound since we parameterize the problem differently: the complexity of our algorithm depends also on the size of the support of the distribution.}
For comparison, we mention that $\sigma \leq \sqrt{\mu(1-\mu)}$ always holds, and when given binary access to the $v_i$, one may additionally assume (after maximum-finding and rescaling) that $\mu \in [1/N,1]$. {The second term in the complexity is then at most $\sqrt{N/\delta} \log(1/\rho)$ (i.e.~at most our bound when we ignore the $\rho$-dependence). The first term, however, is larger than our complexity if and only if $\delta N \leq (\sigma/\mu)^2$ (again ignoring $\rho$). Our algorithm thus improves over prior mean estimation algorithms when the support is relatively small: when $\delta N$ is at most $(\sigma/\mu)^2$.}

\paragraph{Approximate counting algorithms.}

As mentioned above, our algorithm improves the error-dependence for mean estimation (for random variables with small support). It therefore makes sense to compare our upper bound with the well-known lower bound for the approximate counting problem for binary strings $x \in \{0,1\}^N$. We first recall a precise statement.
Let $x \in \{0,1\}^N$ and $k = \abs{x}$, and $U_x$ a unitary implementing quantum oracle access to $x$.
Then for an integer $\Delta > 0$, any quantum algorithm which, with probability $\geq 2/3$, computes an additive $\Delta$-approximation of $k$ uses at least $\Omega(\sqrt{N/\Delta} + \sqrt{k(N-k)} / \Delta)$ applications of controlled-$U_x$~\cite[Thm.~1.10 and 1.11]{nayak1999quantum}.
A matching upper bound is given in~\cite[Thm.~18]{bhmt:countingj}, see~\cref{thm:approx counting} for a precise formulation. We can compare the complexity of our algorithm by converting multiplicative error into additive error, i.e., to achieve an additive error of $\epsilon$ we take $\delta = \eps/k$ (or $\eps$ divided by a suitable multiplicative approximation of $k$).
Then the key point is that if one considers~\cref{eq:informal summing complexity intro} for $\eps \leq \Delta$ and~$k \geq 1$, then
\begin{equation*}
    \sqrt{\frac{Nk}{\eps}} \geq \sqrt{\frac{Nk}{\Delta}} \geq \sqrt{\frac{1}{2} \frac{N}{\Delta} + \frac{1}{2} \frac{k (N-k)}{\Delta}} \geq \frac{1}{2} \left(\sqrt{\frac{N}{\Delta}} + \frac{\sqrt{k (N-k)}}{\Delta}\right)
\end{equation*}
where the last inequality follows from concavity of the square-root function and~$\Delta \geq 1$. In other words, for all parameters $N,k,\Delta$, the complexity of our algorithm (left hand side), is at least as large as the lower bound on approximate counting (right hand side), so we do not break the lower bound.

We highlight two ranges of parameters. On the one hand, when $\Delta$ is large, our upper bound is suboptimal for quantum counting. For example, when $\Delta=k/2$ (i.e.,~$\delta = 1/2$), our algorithm uses $O(\sqrt{N})$ queries whereas the approximate counting algorithm from~\cite[Thm.~18]{bhmt:countingj} uses only $O(\sqrt{N/k})$ queries. This is no surprise given that our algorithm finds all ``large'' elements, which in the counting setting amounts to finding all ones. On the other hand, when $\Delta$ is a small constant, say $\Delta=1$, the approximate counting lower bound shows that our upper bound is essentially tight.
To see this, note if one had a quantum algorithm for computing~$(1 \pm \delta)$-multiplicative approximations of sums with quantum query complexity $O(\sqrt{N}/\delta^c)$ (that succeeds with probability $\geq 2/3$), this would give an upper bound of $O(\sqrt{N} k^c)$ for finding an additive~$\Delta = 1$-approximation of~$k$.
The lower bound becomes $\Omega(\sqrt{k(N-k)})$ when~$\Delta = 1$, and so we must have~$\sqrt{N} k^c \geq \sqrt{k (N-k)}$, i.e., $c \geq 1/2$.
We leave it as an open problem whether one can obtain a quantum algorithm for approximate summing of vectors $v \in [0,1]^N$ that matches the approximate counting complexity when applied to $v \in \{0,1\}^N$ for the entire range of parameters $N,k,\Delta$. 

{Finally we highlight that our quantum upper bound for summing outperforms the classical lower bound for approximate counting for a certain range of parameters. The classical randomized query complexity of achieving a multiplicative $\delta$-approximation of the Hamming weight of $x \in \{0,1\}^N$ is $\Theta(\min\{N,\frac{N}{\delta^2 k}\})$.\footnote{We believe this is well known: the upper bound (which was also conjectured in~\cite{bhmt:countingj}) follows from the algorithm presented in \cite{dagumOptimalAlgorithmMonte2000}, applied to a Bernoulli random variable with success probability $k/N$; the lower bound is claimed for instance in \cite{aaronson&rall:qcounting}, but we could not locate a proof in the literature for cases other than $k = \Theta(N)$ \cite{canettiLowerBoundsSampling1995}. We therefore provide a (not entirely trivial) proof in \cref{sec:classicalLB}.} This classical bound exceeds our quantum upper bound of $\widetilde O(\sqrt{N/\delta})$ if $1/\delta \in O(N)$ and $1/\delta \in \Omega(k^{2/3}/N^{1/3})$ (ignoring logarithmic factors).}

\subsection*{Organization of the paper}
In~\cref{sec:preliminaries} we discuss notation, the computational model, and some basic results we build upon. In~\cref{sec:fastgrover} we provide our algorithm for searching for multiple marked elements. Lastly, in~\cref{sec:fastsum}, we give our summation algorithm.

\section{Preliminaries} \label{sec:preliminaries}
\subsection{Notation and assumptions}
Throughout the paper, we will assume that $N \geq 1$ and $N = 2^n$ for some $n \geq 1$.
We identify $\CC^N$ with $\CC^{2^n}$ by $\ket{j} \mapsto \ket{j_1 \dotsc j_n}$, where $(j_1, \dotsc, j_n) \in \{ 0, 1 \}^n$ is the standard binary encoding of $j-1 \in \{0, \dotsc, 2^n - 1\}$.
We write~$\log$ for the logarithm with base~$2$ and~$\ln$ for the natural logarithm.
For a bit string $x \in \{0,1\}^N$ we write $\abs{x} = \sum_{i \in [N]} x_i$. Throughout we will use $k$ to denote the Hamming weight of $x$, i.e,~$\abs{x}=k$, and we write $\kest,\klb,\kub$ for various bounds on $k$: $\kest$ will denote an integer such that $k/2 \leq \kest \leq 3k/2$, and $\klb$ and $\kub$ are lower- and upper bounds on $k$ respectively.

\subsection{Computational model} \label{sec:computational model}

We express the cost of a quantum algorithm in terms of the number of one- and two-qubit gates it uses.
Note that in particular we allow single-qubit rotations with arbitrary real angles.
In~\cref{sec:fastgrover}, the angle will always be determined by classical data.
In \cref{sec:fastsum} we additionally apply controlled rotations where the control register is allowed to be in superposition; in this case we only use angles of the form $\pi/2^m$ and we carefully count the number of used gates.  In the query setting, we separately count the number of quantum queries the algorithm makes, which means (controlled) applications of the query unitary or its inverse.
We will use the following types of quantum queries to access either $N$-bit strings $x \in \{0,1\}^N$ or $N$-dimensional vectors $v \in [0,1]^N$ (specified in fixed-point format).

\begin{dfn}
\label{dfn:bit string oracle access}
A unitary $U_x \in \U(\CC^N \otimes \CC^2)$ is said to implement quantum oracle access to an $N$-bit string $x \in \{0,1\}^N$ if it acts as
\vspace{-2mm}
\begin{equation*}
\vspace{-2mm}
    U_x \ket{i} \ket{b} = \ket{i} \ket{b \oplus x_i}
\end{equation*}
for all $i \in [N]$ and $b \in \{0,1\}$.
\end{dfn}

\begin{dfn}
\label{dfn:vector oracle access}
A unitary $U_v \in \U(\CC^N \otimes \CC^{2^b})$ is said to implement quantum oracle access to $(0,b)$-fixed-point representations of $v \in [0,1]^N$ if it acts as
\vspace{-2mm}
\begin{equation*}
    \vspace{-2mm}
    U_v \ket{i} \ket{0^b} = \ket{i} \ket{v_i}
\end{equation*}
for all $i \in [N]$, where $\ket{v_i} = \ket{(v_i)_1 \dotsc (v_i)_b}$ satisfies $\sum_{j=1}^b (v_i)_j 2^{-j} = v_i$.
\end{dfn}

In both cases we allow the unitary to act on additional workspace registers, which we omit for notational convenience. Moreover, throughout the paper, every algorithm will use at most logarithmic number of additional ancillary qubits.

We additionally use a classical data structure to maintain sorted lists that supports both insertion and look-up in a time that scales logarithmically with the size of the list, see for example \cite[Sec.~6.2.3]{KnuthVol3} or
\cite[Ch.~13]{cormenIntroductionAlgorithmsFourth2022}.
We emphasize that we allow neither writing nor reading of such a data structure in superposition.

\subsection{Various quantum subroutines}

In this section we summarize the external results that we build upon, and in some cases give a quick proof of an aspect of the result that is not mentioned explicitly in the original source.

\begin{procedure}[t]
  \caption{AmpEst($U,M$)}\label{alg:ampest}
  \Input{Access to controlled versions of unitary $U \in \U(2^q)$ and its inverse, an integer $M \geq 1$.}
  \Output{Real number $\tilde a \in [0,1]$.}
  \Analysis{\cref{lem:ampest}}
\end{procedure}

\begin{lem}[{Amplitude estimation~\cite[Thm.~12]{bhmt:countingj}}]\label{lem:ampest}
Let $U\in \CC^{2^q\times 2^q}$ be a unitary that creates a state
\[
\ket{\psi} = U\ket{0^q}= \sqrt{a}\ket{\phi_1}\ket{1}+\sqrt{1-a}\ket{\phi_0}\ket{0}.
\]
There is a quantum algorithm \AmpEst{} that, with probability $\geq \frac{8}{\pi^2}$, outputs an $\tilde{a} \in [0,1]$ such that
\[
  \abs{a-\tilde{a}} \leq 2\pi \frac{\sqrt{a(1-a)}}{M} + \frac{\pi^2}{M^2}
\]
using $M$ applications of controlled-$U$ and $M$ applications of controlled-$U^{\dagger}$. If~$M$ is a power of~$2$, the algorithm uses $\bigO{qM}$ additional quantum gates, and the computation of the sine-squared function of the normalized phase.
\end{lem}
\begin{proof}
  This follows from the formulation in \cite{bhmt:countingj} by setting $k=1$ and implementing the reflection through $\ket{0^q}$ using $\bigO{q}$ gates, which needs to be performed $M$ times. If $M$ is a power of $2$ we can implement the quantum Fourier transform on $m = \log_2(M)$ qubits using $m$ Hadamard gates, and the QFT and its inverse need only be performed once; therefore, this cost is absorbed in the big-O.
\end{proof}
We note that the above formulation of \AmpEst{} outputs a real number $\tilde a$ whereas we require a fixed-point encoded number for future uses.
However, it suffices to use fixed-point arithmetic using~$\bigO{\log(M)}$ bits; after all, the guarantee of \AmpEst{} only gives a precision of~$1/\!\poly(M)$.

We also need a version of amplitude amplification where the success probability is $1$ if one knows the amplitude of the ``good'' part of the state exactly.
In a nutshell, the algorithm with success probability $1$ is the usual amplitude amplification algorithm applied not to $U$ but to $U$ followed by a rotation of the last qubit to slightly reduce the amplitude $a$ to~$\bar a$.
Carefully choosing $\bar a$ ensures that the success probability is exactly $1$ after an integer number of rounds of amplitude amplification.
This requires having access to gates which implement rotation by arbitrary angles, not just angles of the form~$\pi / 2^m$ for some integer~$m$.
We specialize the statement of this result to the search setting but remark that this works more generally.
For exactly $N/4$ marked elements this observation was first made in \cite{bbht:bounds}.

\begin{procedure}[t]
  \caption{GroverCertainty($U$, $k_0$)}\label{alg:grover certainty}
  \Input{Quantum oracle $U_x$ to access $x \in \{ 0,1 \}^N$, an integer $k_0 \geq 1$.}
  \Output{An index $i \in [N]$.}
  \Guarantee{If $\abs{x} = k_0$, then $x_i = 1$ with certainty.}
  \Analysis{\cref{thm:grover with certainty}}
\end{procedure}

\begin{thm}[{\cite[Thm.~4]{bhmt:countingj}}]
\label{thm:grover with certainty}
Let $x \in \{0,1\}^N$ with $\abs{x} = k \geq 1$.
Then there is a quantum algorithm \GroverCertainty{} that takes as input a quantum oracle~$U_x$ to access~$x$ and an integer $k_0 \in [N]$, and that outputs an index $i \in [N]$, such that $x_i = 1$ with certainty if $k_0 = k$, and uses $O(\sqrt{N/k_0})$ quantum queries to~$x$, and $O(\sqrt{N/k_0} \log(N))$ additional gates.
\end{thm}

The other version of Grover that we need is the following, which is originally due to~\cite[Thm.~3]{bbht:bounds}, but we use a slightly different version from~\cite[Thm.~3]{bhmt:countingj}:
\begin{procedure}[t]
  \caption{GroverExpectation($U_x$)}\label{alg:grover expectation}
  \Input{Quantum oracle $U_x$ to access $x \in \{ 0,1 \}^N$.}
  \Output{An index $i \in [N]$.}
  \Guarantee{If $\abs{x} \geq 1$, then $x_i = 1$ with certainty.}
  \Analysis{\cref{thm:grover in expectation}}
\end{procedure}
\begin{thm}[{\cite[Thm.~3]{bhmt:countingj}}]
\label{thm:grover in expectation}
Let $x \in \{0,1\}^N$ with~$\abs{x} = k$, where~$k$ is not necessarily known.
Then there is a quantum algorithm \GroverExpectation{} that takes as input a quantum oracle $U_x$ to access $x$, and if $k \geq 1$, outputs an index $i \in [N]$ such that $x_i = 1$.
The number of quantum queries to~$x$ that it uses is a random variable $Q$, such that, if $k \geq 1$, then
\vspace{-2mm}
\begin{equation*}
\vspace{-2mm}
    \EE[Q] = \bigO{\sqrt{N/k}},
\end{equation*}
and if $\abs{x} = 0$, then  $Q = \infty$ (i.e., the algorithm runs forever).
The number of additional gates used is $\bigO{Q \, \log(N)}$.
The index $i$ which is output is uniformly random among all such indices, and independent of the value of $Q$.
\end{thm}

\begin{procedure}[t]
  \caption{Grover$_{2/3}$($U_x,\klb$)}\label{alg:grover prob}
  \Input{Quantum oracle $U_x$ to access $x \in \{ 0,1 \}^N$ and a lower bound $\klb$ on $\abs{x}$.}
  \Output{An index $i \in [N]$.}
  \Guarantee{If $\abs{x} \geq 1$, then with probability $\geq 2/3$, $x_i = 1$.}
  \Analysis{\cref{lem:grover with lower bound}}
\end{procedure}

\begin{lem}
\label{lem:grover with lower bound}
Let $x \in \{0,1\}^N$.
Then there is a quantum algorithm \GroverLB{} that takes as input a quantum oracle $U_x$ to access $x$ and a lower bound $\klb \geq 1$ on $\abs{x}$. With probability $\geq 2/3$, it outputs an index $i \in [N]$ such that $x_i = 1$.
It uses~$\bigO{\sqrt{N/\klb}}$ quantum queries to~$x$, and~$\bigO{\sqrt{N/\klb} \log(N)}$ additional gates.
\end{lem}
\begin{proof}
The algorithm \GroverExpectation{} finds an index $i$ such that $x_i=1$. Its number of applications of controlled-$U_x$ is a random variable $Q$ and the number of additional gates is $O(Q\cdot \log(N))$. By \cref{thm:grover in expectation} we have $\EE[Q] = O(\sqrt{N/\abs{x}})$. Markov's inequality shows that if we terminate \GroverExpectation{} after at most $C\sqrt{N/\abs{x}}$ quantum queries for a suitable constant $C>0$, then it finds an index $i$ such that $x_i=1$ with probability at least $2/3$. The procedure \GroverLB{} uses the lower bound $\klb$ on $\abs{x}$ to decide when to terminate \GroverExpectation{}. For the same constant $C>0$ as before, it terminates after at most $C \sqrt{N/\klb}$ quantum queries.
Since $C\sqrt{N/\klb} \geq C \sqrt{N/\abs{x}}$, the success probability of \GroverLB{} is also at least $2/3$.
\end{proof}

Let us make some remarks about the complexity of finding a single marked element. First, to find such an element with certainty one can essentially remove the $\log(N)$ factor in the gate complexity: $O(\sqrt{N}\log(\log^*(N)))$ gates suffice~\cite{arunachalam:qsearchgatesj}. Second, by cleverly combining \GroverCertainty{} and \GroverLB{}, one can find a marked element (among an unknown number of solutions) with probability $\geq 1-\rho$ using $\sqrt{N \log(1/\rho)}$ quantum queries~\cite{buhrmanBoundsSmallerrorZeroerror1999}. This shows that the standard way of boosting the success probability of \GroverLB{} is not optimal.

Next, we recall a well-known result on approximate counting.
\begin{procedure}[t]
  \caption{ApproxCount($U_x$, $\eps$, $\rho$)}\label{alg:approx counting mult error}
  \Input{Quantum oracle $U_x$ to access $x \in \{ 0,1 \}^N$, rational number $\eps > 0$ such that $\frac{1}{3N} < \eps \leq 1$, failure probability $\rho > 0$.}
  \Output{Integer $\tilde k \in \{ 0, \dotsc, N \}$.}
  \Guarantee{If $\abs{x} = k \geq 1$, with probability $\geq 1 - \rho$, $\abs{\tilde k - k} \leq \eps k$, and if $k = 0$ then $\tilde k = 0$ with certainty.}
  \Analysis{\cref{thm:approx counting}}
\end{procedure}

\begin{thm}[{\cite[Thm.~18]{bhmt:countingj}}]
\label{thm:approx counting}
Let $x \in \{0,1\}^N$ and write $\abs{x}=k$. Let $\frac{1}{3N}<\eps \leq 1$.
Then there is a quantum algorithm that, with probability at least~$2/3$, that outputs an estimate $\tilde k$ such that
\vspace{-2mm}
\[
\vspace{-2mm}
\abs{\tilde k - k} \leq \eps k
\]
using an expected number of
\[
    \Theta\left(\sqrt{\frac{N}{\lfloor \eps k \rfloor +1}} + \frac{\sqrt{k(N-k)}}{\lfloor \eps k\rfloor +1}\right)
\]
quantum queries to~$x$. If $k=0$, then the algorithm outputs $\tilde k = 0$ with certainty, using $\Theta(\sqrt{N})$ quantum queries to $x$. In both cases, the algorithm uses a number of gates which is~$\bigO{\log(N)}$ times the number of quantum queries.
 To boost the success probability to $1-\rho$, repeat the procedure $O(\log(1/\rho))$ many times and output the median of the returned values.
\end{thm}
We often use the special case $\eps=1/2$ of the above theorem, hence we record it here for future use. (Note that the proof of \cref{thm:approx counting} given in \cite{bhmt:countingj} in fact starts by obtaining a constant factor approximation of $\abs{x}$.)
\begin{cor} \label{cor:factor3-approx}
Let $x \in \{0,1\}^N$ and write $\abs{x}=k$. Then there is a quantum algorithm that outputs a $\kest$ such that, with probability $\geq 1-\rho$, we have $k/2\leq \kest \leq 3 k/2$, and uses $\bigO{\sqrt{N/({k+1})}\log(1/\rho)}$ quantum queries and $\bigO{\sqrt{N/({k+1})}\log(1/\rho)\log(N)}$ gates.
\end{cor}
We now discuss known extensions of the above results on counting the Hamming weight of a bit string to the problem of \emph{mean estimation}: given a vector $v \in [0,1]^N$, one is interested in approximating $\bar{v} = \frac{1}{N} \sum_{i=1}^N v_i$.
This was first studied in~\cite{grover:telecomputation} and later in \cite{grover1998framework} where in the latter it was shown that one can find an additive $\eps$-approximation of $\bar{v}$ using $\bigOt{1/\eps}$ quantum queries to a unitary that prepares a state encoding the entries of~$v$ in its amplitudes, and a similar number of additional gates (also dependent on~$N$).
Using amplitude amplification techniques one can reduce the query dependence to $O(1/\eps)$ with $O(\log(N)/\eps)$ additional gates.
This result may be easily recovered from \cref{lem:ampest} with~$M = \Theta(1/\eps)$, applied to a unitary preparing
\vspace{-1mm}
\begin{equation*}
    \vspace{-1mm}
    \frac{1}{\sqrt{N}} \sum_{i=1}^N  \ket{i} (\sqrt{1 - v_i} \ket{0} + \sqrt{v_i} \ket{1}).
\end{equation*}
It is well-known that when one has quantum oracle access to fixed point representations of the entries of~$v$ (cf.~\cref{dfn:vector oracle access}), rather than just a state encoding its entries in the amplitudes, one can give an algorithm whose complexity depends only on~$N$ and~$\delta$, with guarantees as given below.
\begin{thm}
\label{thm:normal multiplicative mean est}
Let $v \in [0,1]^N$ be a vector with each entry $v_i$ encoded in $(0,b)$-fixed-point format, and let~$U_v$ be a unitary implementing binary oracle access to~$v$ (cf.~\cref{dfn:vector oracle access}).
Let~$\rho \in (0,1)$.
Then with $\bigO{\frac{\sqrt{N}}{\delta} \log(1/\rho)}$ applications of controlled-$U_v$, controlled-$U_v^\dagger$, and a polylogarithmic gate overhead,
one can find with probability~$\geq 1 - \rho$ a multiplicative $\delta$-approximation of $\frac{1}{N} \sum_i v_i$.
\end{thm}
We give an informal description of the algorithm here, and refer the interested reader to~\cite{qscalingICALP} for a careful implementation along with a bit complexity analysis.
By using quantum maximum finding~\cite{durr&hoyer:minimum}, with $\bigO{\sqrt{N}}$ quantum queries and $O(b\sqrt{N}\log(N))$ other gates, one may find $v_{\max} = \max_i v_i$.
If $v_{\max} = 0$ one may output $\bar{v}_\est = 0$ as an estimate of $\bar{v}$.
Note that having binary access here makes it easy to compare elements.
Next, set $w_i = v_i / v_{\max}$, and let~$U$ be a unitary preparing a state
\vspace{-2mm}
\begin{equation*}
\vspace{-2mm}
    \frac{1}{\sqrt{N}} \sum_{i=1}^N  \ket{i} (\sqrt{1 - w_i} \ket{0} + \sqrt{w_i} \ket{1}).
\end{equation*}
Then~\cref{lem:ampest} with~$M = 8 \sqrt{N}/\delta$ outputs an estimate~$\bar{w}_\est$ of~$\bar{w}$, such that
\begin{equation*}
    \abs{\bar{w}_\est - \bar{w}} \leq 2\pi \frac{\sqrt{\bar{w} (1 - \bar{w})}}{M} + \frac{\pi^2}{M^2} \leq \frac{2 \pi}{8} \delta \bar{w} \sqrt{1 - \bar{w}} + \frac{\pi^2}{64} \delta^2 \bar{w} \leq \delta \bar{w}
\end{equation*}
because $1/N \leq \bar{w}$, so $\bar{w}_\est$ is a multiplicative $\delta$-approximation of $\bar{w}$.
Therefore $\bar{w}_\est \cdot v_{\max}$ is a multiplicative $\delta$-approximation of $\bar{v}$.
We note that in this step the binary access to the entries of~$v$ enables the ``binary amplification'' by ensuring the largest entry of~$w$ is~$1$.

\section{Fast Grover search for multiple items, without quantum memory}
\label{sec:fastgrover}

In this section we give a version of Grover's search algorithm for the problem of, given a string~$x \in \{0,1\}^N$, finding \emph{all}~$k$ indices~$i \in [N]$ such that~$x_i = 1$.
For~$\rho \in (0,1)$ with~$\rho = \Omega(1/\!\poly(k))$, our algorithm finds all such indices with probability~$\geq 1-\rho$, and uses~$\bigO{\sqrt{Nk}}$ quantum queries and~$\bigOt{\sqrt{Nk}}$ single- and two-qubit gates.
The contribution here is that the query complexity is optimal and the time complexity is only polylogarithmically worse than the query complexity, without using a QRAM.

\subsection{Deterministic Grover for multiple elements}
\label{subsec:grovercertaintymultiple}
We first recall the well-known result~\cite[Lem.~2]{degraafQuantumVersionsYao2002}, that it is possible to find all solutions with probability~$1$ using~$\bigO{\sqrt{Nk}}$ quantum queries, which is optimal, but suffers from a too-high gate complexity in terms of $k$.
The algorithm is given in \GroverCertaintyMultiple{}.
We first define for each $j \in [N]$ a gate $C_j$, referred to as the ``control-on-$j$-NOT''-gate, and describe how to implement it with a standard gate-set.
The point of this gate is that if one has quantum oracle access $U_x$ to $x \in \{0,1\}^N$, then $C_j U_x$ implements quantum oracle access to the bit string $y \in \{0,1\}^N$ which agrees with $x$ on all indices, except on the $j$-th index, where the bit is flipped.
\begin{lem}
\label{lem:C_j gate properties}
Let $N = 2^n$.
For $j \in [N]$ define the ``control-on-$j$-NOT''-gate $C_j \in \U(\CC^N \otimes \CC^2)$ by
\begin{equation}
\label{eq:C_j gate definition}
    C_j \ket{i} \ket{b} = \begin{cases}
      \ket{i} \ket{b \oplus 1} & \text{if } i = j, \\
      \ket{i} \ket{b} & \text{otherwise.}
    \end{cases}
\end{equation}
Then the $C_j$-gate can be implemented with $\bigO{n}$ standard gates and $n - 1$ ancillary qubits.
\end{lem}
\begin{proof}
Let $\ket{j} = \ket{j_1 \dotsc j_n}$ be the binary encoding of $j-1$.
Then:
\begin{enumerate}
    \item For each $l \in [n]$ such that $j_l = 0$, apply a NOT gate on the $l$-th qubit of the index register.
    \item Apply a NOT-gate to the output register containing $b$, controlled on all $n$ qubits of the index register.
    This can be implemented using $O(n)$ Toffoli gates, one CNOT gate, and $n-1$ ancilla qubits, see \cite[Fig.~4.10]{nielsen-chuang}.
    \item Apply the NOT gates from the first step again. \qedhere
\end{enumerate}
\end{proof}

\begin{procedure}[t]
  \caption{GroverCertaintyMultiple($U_x$, $\kub$)}\label{alg:grover multi certainty large gate overhead}
  \Input{Quantum oracle $U_x$ to access $x \in \{0,1\}^N$, an integer $\kub \geq 1$.}
  \Output{Classical list of indices $J \subseteq [N]$.}
  \Guarantee{If $\abs{x} \leq \kub$, then for every $j \in [N]$, $j \in J$ if and only if $x_j = 1$.}
  \Analysis{\cref{lem:all-sol-grover-slow}}
  \medskip

  $J_{\kub} \leftarrow \emptyset$\;
  $U_{J_{\kub}} \leftarrow U_x$\;
  $m \leftarrow \kub$\;

  \smallskip
  \While{$m > 0$}{
    use \GroverCertainty$(U_{J_m}, m)$ to find a $j \in [N] \setminus J_m$\;
    \uIf{$x_j = 1$}{
        $J_{m-1} \leftarrow J_m \cup \{j\}$\;
        $U_{J_{m-1}} \leftarrow C_j U_{J_m}$, where $C_j$ is defined in \cref{lem:C_j gate properties}\;
    } \Else{
        $J_{m-1} \leftarrow J_m$\;
        $U_{J_{m-1}} \leftarrow U_{J_m}$\;
    }
    $m \leftarrow m - 1$\;
  }

  \Return{$J_0$}\;
\end{procedure}
\begin{lem}\label{lem:all-sol-grover-slow}
    Let $x\in \{0,1\}^N$, $U_x$ a quantum oracle to access $x$, and $\kub \geq 1$. If $\abs{x} = k \leq \kub$, then \GroverCertaintyMultiple{}($U_x$, $\kub$) finds, with probability $1$, all $k$ indices $i$ such that $x_i = 1$.
    The algorithm uses
    \[
        \bigO{\sqrt{N \kub}}
    \]
    applications of $U_x$, and
    \[
    \bigO{\sqrt{N \kub} (k + 1) \log(N)}
    \]
    additional non-query gates.
\end{lem}
\begin{proof}
    We first establish correctness of \GroverCertaintyMultiple{}.
    For $m \in [\kub]$, let $J_m \subseteq [N]$ be the index set and $U_{J_m}$ the unitary in the algorithm at the $m$-th step.
    Then by the definition of $C_j$, $U_{J_m}$ implements oracle access to the bit string $y^m$ which agrees with $x$ on $[N] \setminus J_m$, and is zero on the indices in $J_m$ (whereas $x_j = 1$ for $j \in J_m$). Clearly $j \in J_0$ implies that $x_j=1$. It remains to show that in $\kub$ iterations we find \emph{all} marked elements. To do so, observe that there can be at most $\kub - k$ iterations in which one fails to find a new $j \in [N]$ such that $x_j = 1$: indeed, as soon as this happens, we have $m = \abs{y^m}$, and every iteration afterwards we find a new index with certainty by the guarantees of \GroverCertainty{}.

    In total, this procedure uses $\sum_{m=1}^{\kub} \bigO{\sqrt{N/m}} = \bigO{\sqrt{N\kub}}$ applications of $U_x$.
    The number of auxiliary gates for a single query in the $m$-th iteration is $\bigO{\abs{J_m} \cdot \log(N)}$, and \GroverCertainty{} itself uses an additional $\bigO{\sqrt{N/m} \log(N)}$ additional gates.
    Therefore the total number of gates in the $m$-th iteration is
    \begin{equation*}
        \bigO{\sqrt{N/m} \cdot \abs{J_m} \cdot \log(N) + \sqrt{N/m} \log(N)} = \bigO{\sqrt{N/m} (k + 1) \log(N)}
    \end{equation*}
    Summing this over all iterations yields a total gate complexity of $\bigO{\sqrt{N\kub} (k + 1) \log(N)}$.
\end{proof}


\subsection{Coupon collecting Grover}
\label{subsec:grovercoupon}
We next give another simple version of Grover which can be used to find a large fraction of the marked elements in a time-efficient manner, but does not yield a query-optimal bound when the fraction is close to~$1$.
The algorithm is given in~\GroverCoupon{}, and is analyzed in~\cref{prop:most-sol-grover}.

The algorithm is simple: the idea is to repeatedly call~\GroverLB{} to sample marked elements.
The analysis is based on the observation that the required number of calls to~\GroverLB{} is a sum of geometrically distributed random variables: for~$1 \leq i \leq t$, the number of calls to obtain the~$i$-th distinct marked element is a geometrically distributed random variable with success probability~$p_i' \geq \frac23 p_i$, where~$p_i = (k-i+1)/k$ is the probability of observing a new element after $i-1$ distinct elements have been found.
This is because~$\GroverLB{}$ succeeds with probability~$\geq 2/3$, and the fact that if~$\GroverLB{}$ successfully finds a marked element, then it is uniformly random among the marked elements. The number of calls can then by bounded using a general tail bound on sums of geometrically distributed random variables given in~\cite[Thm.~2.3]{Janson:geometrictail} (see \cref{lem:geometrictail}).

The analysis is based on tail bounds of sums of geometrically distributed random variables. These tail bounds in turn are stated in terms of the harmonic numbers, for which we recall some basic properties in the following lemma.
\begin{lem}
\label{lem:harmonic numbers}
The $k$-th harmonic number $H_k$ is defined by $H_k = \sum_{j=1}^k \frac{1}{j}$, and we shall use the convention $H_0 = 0$.
For $k \geq 1$ it satisfies
\begin{equation*}
    H_k - \gamma - \ln(k) \in \left[\frac{1}{2(k+1)}, \frac{1}{2k}\right],
\end{equation*}
where $\gamma \approx 0.577$ is the Euler--Mascheroni constant.
Furthermore, for $0 \leq t < k$, this implies
\begin{equation*}
    H_k - H_{k-t} \leq \ln\left(\frac{k}{k-t}\right) + \frac{2k-t+1}{2 k (k-t+1)},
\end{equation*}
which in turn for $t \leq k/2 < k$ implies
\begin{equation*}
    H_k - H_{k-t}
    \leq \frac{2(t+1)}{k}.
\end{equation*}
\end{lem}
\begin{proof}
  The bounds on $H_k - \gamma - \ln(k)$ are well-known, see~\cite{young1991} for an elementary proof.
    For the estimate on $H_k - H_{k-t}$ we have
    \begin{align*}
        H_k - H_{k-t} \leq \ln(k) - \ln(k-t) + \frac{1}{2k} + \frac{1}{2(k-t+1)} = \ln \left(\frac{k}{k-t}\right) + \frac{2k-t+1}{2 k (k-t+1)}.
    \end{align*}
    Furthermore, if $t \leq k/2 < k$, then
    \begin{align*}
    \ln\left(1 + \frac{t}{k-t}\right) + \frac{2k-t+1}{2 k (k-t + 1)} & \leq \frac{2t}{k} + \frac{2k+1}{2k (k/2 + 1)}
    \leq \frac{2t}{k} + \frac{2}{k} = \frac{2(t+1)}{k}. \qedhere
    \end{align*}
\end{proof}

We use the following tail bound for geometrically distributed variables.
\begin{lem}[{\cite[Thm.~2.3]{Janson:geometrictail}}] \label{lem:geometrictail}
  For $i \in [n]$ assume $X_i \sim \mathrm{Geo}(p_i)$ for $p_i \in (0,1]$. Let $X = \sum_{i \in [n]} X_i$ and write $\mu = \Exp[X]$, $p_* = \min_{i \in [n]} p_i$. Then for any $\lambda \geq 1$ we have
  \[
    \Pr[X \geq \lambda \mu] \leq \lambda^{-1}(1-p_*)^{(\lambda-1-\ln(\lambda))\mu}.
  \]
\end{lem}

\begin{cor} \label{cor:geometrictail}
    For $i \in [n]$ assume $X_i \sim \mathrm{Geo}(p_i)$ for $p_i \in (0,1]$. Let $X = \sum_{i \in [n]} X_i$ and write $\mu = \Exp[X]$, $p_* = \min_{i \in [n]}  p_i$. Let $\rho \in (0,1)$. Then~$\Pr[X \geq T] \leq \rho$ whenever
    \[
    T \geq 2 \ln(2) \mu + 2 \frac{\ln(1/\rho)}{\ln(1/(1-p_*))}.
    \]
\end{cor}
\begin{proof}
    We apply \cref{lem:geometrictail} with $\lambda \geq 1$ to obtain
\[
\Pr[X \geq \lambda \mu] \leq \lambda^{-1}(1-p_*)^{(\lambda-1-\ln(\lambda))\mu} \leq (1-p_*)^{(\lambda-1-\ln(\lambda))\mu}.
\]
By the first-order characterization of convexity of~$\lambda \mapsto \lambda - 1 - \ln(\lambda)$ at~$\lambda = 2$, we have
\begin{equation*}
    \lambda - 1 - \ln(\lambda) \geq \left(1 - \frac{1}{2}\right) (\lambda - 2) + 2 - 1 - \ln(2) = \frac{1}{2} \lambda - \ln(2).
\end{equation*}
Therefore
\begin{align*}
\Pr[X \geq \lambda \mu ] & \leq e^{\ln(1-p_*) (\lambda/2 - \ln(2)) \mu},
\end{align*}
and so to ensure that this is at most~$\rho$, it suffices to take
\begin{equation*}
    \lambda \mu \geq 2 \ln(2) \mu + \frac{2 \ln(\rho)}{\ln(1 - p_*)}.
\end{equation*}
Note that such~$\lambda$ also satisfies $\lambda \geq 1$, because~$2 \ln(2) \geq 1$ and~$\ln(\rho) / \ln(1 - p_*) \geq 0$.
Therefore we have shown that~$\Pr[X \geq T] \leq \rho$ whenever~$T \geq 2 \ln(2) \mu + 2 \frac{\ln(\rho)}{\ln(1 - p_*)}$.
\end{proof}

Applying the above tail bound with~$p_i \geq \tfrac23 (k - i + 1) / k$ yields the following lemma.
\begin{lem} \label{lem:GroverCouponSampleComplexity}
Let~$1 \leq t \leq k \leq N$ and subset~$I \subseteq [N]$ of size~$k$, and let~$\rho \in (0,1)$.
Consider a procedure in which at each step with probability~$\geq 2/3$, one obtains a uniformly random sample from~$I$.
The outputs of
\[
r \geq 3 \ln(2) \, k(H_k - H_{k-t}) + \frac{2 \ln(1/\rho)}{\ln(3k / (k + 2(t-1)))} =: R_{t,k,\rho}
\]
repetitions of this procedure suffice to, with probability $\geq 1-\rho$, obtain $t$ distinct samples from~$I$.
\end{lem}
\begin{proof}
    We apply \cref{cor:geometrictail} to $X = \sum_{i \in [t]} X_i$ where $X_i \sim \mathrm{Geo}(p_i)$ with $p_i \geq \tfrac23 (k - i + 1) / k$. We then have $p_* \coloneqq \min_{i \in [t]} p_i \geq \frac{2}{3}(k-t+1)/k$, and thus $\frac{1}{1-p_*} \leq \frac{1}{1-\frac{2}{3}(k-t+1)/k} = \frac{3k}{k+2(t-1)}$, and $\mu \coloneqq \EE[X] = \sum_{i \in [t]} \frac{1}{p_i} \leq \sum_{i \in [t]} \frac{3}{2} \frac{k}{k-t+1} = \frac{3}{2} k (H_k - H_{k-t})$. \cref{cor:geometrictail} thus shows that we obtain $t$ distinct samples from $I$ with probability at least $1-\rho$, whenever we use at least $R_{t,k,\rho}$ repetitions of this procedure. 
\end{proof}
We briefly emphasize the value of this lemma.
For general~$t \leq k$, we can use the simple bound $\ln(k/(t-1)) \geq \ln(k/(k-1)) \geq 1/k$ and the estimate $H_k-H_{k-t} \approx \ln(k/(k-t))$, to obtain that $r \in \Omega(k\log(k) + k \ln(1/\rho)) = \Omega(k\log(k/\rho))$ samples suffice.
By contrast, an application of Markov's inequality only yields a sample complexity upper bound of~$k\log(k) \log(1/\rho)$.
In later applications (cf.~\cref{thm:all-sol-grover-fast}), we apply this with $t$ at most $k/2$, in which case we can give tighter estimates.
Indeed, the factor~$1/\ln(3k/(k+2(t-1)))$ is then at most a constant and $H_k - H_{k-t} \leq \frac{2(t+1)}{k}$ by \cref{lem:harmonic numbers}, and thus $r \in \Omega(t + \ln(1/\rho))$ samples suffice.
Therefore the bound is an improvement over the sample complexity of~$\Omega(t \ln(1/\rho))$ one would obtain from a simple application of Markov's inequality -- in particular, one can now ``for free'' choose~$\rho$ to be exponentially small in~$t$ (and similar above).

By using~\GroverLB{} to obtain the samples required for~\cref{lem:GroverCouponSampleComplexity}, we obtain the following algorithmic result.

\begin{procedure}[t]
  \caption{GroverCoupon($U_x$, $R$, $\klb$, $t$)}\label{alg:grover coupon collector}
  \Input{Quantum oracle $U_x$ to access $x \in \{0,1\}^N$, an integer~$R \geq 1$, an integer~$\klb$ such that~$\klb \leq \abs{x}$, an integer $t \geq 1$ such that $1 \leq t \leq \abs{x}$.}
  \Output{Classical sorted list of indices $J \subseteq [N]$.}
  \Guarantee{If $R \geq R_{t,k,\rho} = 3 \ln(2) k (H_k - H_{k-t}) + 2 \frac{\ln(1/\rho)}{\ln(3k/(k+2(t-1))}$, then, with probability~$\geq 1 - \rho$, we have $\abs{J} =t$ and $x_j =1$ for all $j \in J$.}
  \Analysis{\cref{prop:most-sol-grover}}
  \medskip

  $J \leftarrow \emptyset$\;

  \smallskip
  \For{$r=1,\ldots,R$}{
    use \GroverLB{} with arguments ($U_x$, $\klb$) to find a $j \in [N]$ such that $x_j = 1$ with probability $\geq 2/3$ \label{line:grover application}\;
    \If{$j\not\in J$ and $x_j = 1$}{
        add $j$ to $J$\;
    }
    \If{$\abs{J} = t$ \label{line:sizeJ}}{
        \Return{$J$} \;
    \label{line:endfor}}
  }
  \Return{$J$}\;
\end{procedure}

\begin{prp}
\label{prop:most-sol-grover}
    Let $x\in \{0,1\}^N$ with $\abs{x} = k$ unknown, let~$R \geq 1$,
    let~$\klb \geq 1$ be such that~$\klb \leq k$, let $t \geq 1$, and $\rho \in (0,1)$.
    Assume $1 \leq t \leq k$.
    Then \GroverCoupon{} called with a quantum oracle $U_x$ to access $x$, and additional inputs $R$, $\klb$, and $t$, uses $\bigO{\sqrt{N/\klb} \, r}$ quantum queries to~$x$ and~$\bigO{\sqrt{N/\klb} \, r \log(N)}$ additional quantum gates.
    Here, $r$ is a random variable such that~$r \leq R$ with certainty, and with probability $\geq 1 - \rho$, one has
    \begin{equation*}
        r \leq R_{t,k,\rho} = 3 \ln(2) k (H_k - H_{k-t}) + 2 \frac{\ln(1/\rho)}{\ln(3k/(k+2(t-1))}.
    \end{equation*}
    If~$R \geq R_{t,k,\rho}$, then with probability $\geq 1-\rho$, it finds a set of $t$ distinct marked elements, uniformly at random from the set of all sets of $k$ marked elements.
\end{prp}
\begin{proof}
    We first analyze the complexity of \GroverCoupon{}. Let $r \in [R]$ be the number of times the algorithm repeats \cref{line:grover application} through \cref{line:endfor}.
    By \cref{lem:grover with lower bound}, the application of \GroverLB{} in \cref{line:grover application} uses $O(\sqrt{N/\klb})$ quantum queries and $O(\sqrt{N/\klb}\log(N))$ additional gates.
    With one additional query we can verify if the index $j \in [N]$ that is returned by \GroverLB{} is such that $x_j=1$. If indeed $x_j=1$, then we add $j$ to $J$.
    As mentioned in \cref{sec:computational model},  we can insert an element in the sorted list~$J$ in (classical) time $O(\log(N))$.
    We can verify \cref{line:sizeJ} in time $O(\log(N))$ by maintaining a counter for $\abs{J}$.
    The above shows that \GroverCoupon{} indeed uses $\bigO{\sqrt{N/\klb} \, r}$ quantum queries and~$\bigO{\sqrt{N/\klb} \, r \log(N)}$ additional quantum gates.

    We now establish correctness. By construction $r \leq R$ with certainty. \Cref{lem:grover with lower bound} shows that, with probability $\geq 2/3$, the index returned by \GroverLB{} in \cref{line:grover application} is a uniformly random marked element. Hence \cref{lem:GroverCouponSampleComplexity} shows that after obtaining
    \[
        R_{t,k,\rho} = 3 \ln(2) \, k(H_k - H_{k-t}) + \frac{2 \ln(1/\rho)}{\ln(3k / (k + 2(t-1)))}
    \]
    such indices, we have obtained $t$ distinct indices with probability at least $1-\rho$. In other words, if $R \geq R_{t,k,\rho}$, then, with probability at least $1-\rho$, \GroverCoupon{} terminates at \cref{line:endfor} with a sorted list $J \subseteq [N]$ of $t$ distinct marked indices.
\end{proof}

\subsection{Grover for multiple elements, fast}
\label{subsec:grovermultiplefast}
In this section we improve the complexity of finding all marked indices by combining the two previously discussed algorithms, \GroverCoupon{} and \GroverCertaintyMultiple{}. The structure of our algorithm, \GroverMultipleFast{}, is as follows. As before, suppose we are given query access to an $x \in \{0,1\}^N$. Let the (unknown) number of marked indices be $k \geq 1$, i.e., $k=\abs{x}$. We first use \GroverCoupon{} to find a (large) fraction of the marked elements. That is, we find a uniformly random subset $J_0 \subseteq [N]$ of $\tau k$ marked elements, where $0<\tau<1$ is a parameter we can use to tune the complexity of the algorithm. This subset $J_0$ partitions $[N]$ into intervals. We then use \GroverCertaintyMultiple{} to find all marked indices in each interval separately.

The following lemma upper bounds the probability that when we draw a set $S \subseteq [k]$ of size $t$ uniformly at random, there exists an interval of length $\geq \ell$ in the set $[k] \setminus S$.
In the analysis of \GroverMultipleFast{} (see \cref{thm:all-sol-grover-fast}), we will use this bound to control the number of elements that are in between any two elements of the previously sampled indices~$J_0$.
\begin{lem}\label{lem:prob-of-run}
  Let $S \subseteq [k]$ be a uniformly random $t$-element set, and let~$1 \leq \ell \leq k - t$.
  The probability that~$[k] \setminus S$ contains a contiguous subset~$I$ of length~$\geq \ell$, i.e., $I = \{a, a+1, \dotsc, a + \ell - 1\}$ for~$1 \leq a \leq k-\ell+1$, is at most~$(k - \ell + 1) (1 - \frac{t}{k})^\ell$.
\end{lem}
\begin{proof}
The probability that~$[k] \setminus S$ contains a contiguous subset~$I$ of length at least~$\ell$ is the same as the probability that it contains a contiguous subset of length \emph{exactly} $\ell$.
This is in turn given by
\begin{align*}
    \Pr[\exists\, a \in \{ 1, \dotsc, k - \ell + 1 \} : \{ a, \dotsc, a + \ell - 1 \} \cap S = \emptyset].
\end{align*}
By a union bound, this is at most
\begin{align*}
    \sum_{a = 1}^{k-\ell + 1} \Pr[\{ a, \dotsc, a + \ell - 1 \} \cap S = \emptyset].
\end{align*}
By the uniform randomness of~$S$, each of the latter probabilities is the same, and given by
\begin{align*}
    \Pr[\{ a, \dotsc, a + \ell - 1 \} \cap S = \emptyset] & = \frac{\binom{k-\ell}{t}}{\binom{k}{t}} = \frac{(k-t)(k-t-1)\cdots (k-t-l+1)}{k(k-1)\cdots(k-l+1)} \leq \left(\frac{k-t}{k}\right)^\ell.
\end{align*}
We conclude that the probability that~$[k] \setminus S$ contains a contiguous subset~$I$ of at least~$\ell$ is at most~$(k - \ell + 1) (1 - \frac{t}{k})^\ell$.
\end{proof}

\begin{procedure}[ht]
  \caption{GroverMultipleFast($U_x$, $\kest$, $\rho$, $\lambda$)}\label{alg:grover fast}
  \Input{Quantum oracle $U_x$ to access $x \in \{0,1\}^N$, an integer $\kest \geq 1$ such that $\abs{x}/2 \leq \kest \leq 3\abs{x}/2$, a failure probability $\rho > 0$, threshold parameter $\lambda \in [6, \kest]$.} 
  \Output{Classical list of indices $J \subseteq [N]$.}
  \Guarantee{If $\lambda$ and $\rho$ are such that $\log(6\kest/\rho) \leq \lceil \kest/\lambda\rceil$, then, with probability $\geq 1-\rho$, we have $\abs{J} = \abs{x}$ and $x_j = 1$ for all $j \in J$.} 
  \Analysis{\Cref{thm:all-sol-grover-fast}}
  \medskip

  $J \leftarrow \emptyset$\;

  $t \leftarrow \lceil k/\lambda \rceil$\;
  $R \leftarrow 6 \ln(2) (t+1) + 2 \ln(1/\rho) \ln(3/2)$\;
  use \GroverCoupon{}($U_x, R, \frac{2}{3}\kest,t$) to find, with probability $\geq 1 - \rho/3$, a sorted list $J_0 \subseteq [N]$ with $x_j = 1$ for all $j \in J_0$, $\abs{J_0} = t$\;
  set $J \leftarrow J_0$ and write 
  $J_0 = \{ a_1 < a_2 < \dotsb < a_t \}$\;
  set $a_0 = 0$ and $a_{t+1} = N+1$\;

  \smallskip
  \For{$i = 0, \dotsc, t$}{
    If $a_{i+1}=a_i+1$, continue with next loop; otherwise, let $b_i = 2^{\lceil \log(a_{i+1}-1-a_i)\rceil}$\;
    construct from $U_x$ an oracle $U_y$ which implements access to the bit string $y \in \{0,1\}^{b_i}$ given by $y_j = x_{a_i + j}$ if $a_i + j < a_{i+1}$, and $0$ otherwise\;

    $(k_j)_{\mathrm{est}} \leftarrow $ \ApproxCount{}($U_y$, $\frac12$, $\frac{\rho}{3(t+1)}$)\;
    \label{algline:find remaining in small interval}
    use \GroverCertaintyMultiple{}($U_y$, $2 (k_j)_{\mathrm{est}}$) to find all $j \in (a_i, a_{i+1})$ such that $x_j = 1$, and add these to $J$\;
  }
  \Return{$J$}\;
\end{procedure}

\begin{thm}\label{thm:all-sol-grover-fast}
Let $x\in \{0,1\}^N$ with~$\abs{x} = k \geq 2$, and assume one knows~$\kest \geq 1$ such that $k/2 \leq \kest \leq 3k/2$. Let $0 < \rho < 1$ and $6\leq \lambda \leq  \kest$ be such that~$t := \lceil \kest / \lambda \rceil \geq \log(6 \kest / \rho)$. Then
    \[
        \bigO{\sqrt{Nk} \left(1 + \frac{1}{\sqrt{\lambda}} \log(k/\rho \lambda) \right) }
    \]
quantum queries to $x$ suffice to, with probability $\geq 1-\rho$, find all $k$ indices $i$ s.t.\ $x_i = 1$. The algorithm uses an additional
\[
 \bigO{\sqrt{Nk}\lambda\log(k/\rho)\log(N)} 
\]
non-query gates.
\end{thm}

We remark here that~\GroverMultipleFast{} takes a multiplicative estimate~$\kest$ of~$k$ as additional input, which can be found with~$\bigO{\sqrt{N/k} \log(1/\rho)}$ quantum queries and~$\bigO{\sqrt{N/k} \log(1/\rho) \log(N)}$ additional gates; see~\cref{cor:factor3-approx}.
Both of these costs are dominated by that of finding the actual elements.
The above theorem also includes a parameter $\lambda$ that allows for a trade-off between query complexity and gate complexity. Before we provide the proof of \cref{thm:all-sol-grover-fast},
let us highlight the two extremal cases that follow from taking $\lambda$ either as large as useful or as small as possible.
\begin{cor}
  Let $x\in \{0,1\}^N$ with $\abs{x}=k\geq 2$. Assume one knows $\kest$ such that  $k/2 \leq \kest \leq 3k/2$. Let $1 > \rho >0$. Then we can find, with probability $\geq 1-\rho$, all $k$ indices $i$ for which $x_i=1$ using either:
\begin{itemize}
    \item $\bigO{\sqrt{Nk}}$ quantum queries and time complexity
 $\bigO{\sqrt{Nk} \min\{\log^3(k/\rho), k\} \log(N)}$, via \cref{thm:all-sol-grover-fast} with $\lambda = \min\{\kest/\log(6 \kest/\rho),  \log^2(\kest/\rho)\}$,\footnote{Strictly speaking, this choice of $\lambda$ could be smaller than $6$, but in that case \GroverCertaintyMultiple{} already has the stated complexity.}  or,
\item $\bigO{\sqrt{Nk}\log(k/\rho)}$ quantum queries and time complexity $\bigO{\sqrt{Nk} \log(k/\rho)\log(N)}$, via  \cref{thm:all-sol-grover-fast} with $\lambda = 6$.
\end{itemize}
\end{cor}

\begin{proof}[Proof of \cref{thm:all-sol-grover-fast}]
Let $t = \lceil \kest / \lambda \rceil$.
Note that because $\lambda \geq 6$ and $\kest \leq 3 k / 2$, we have $t \leq k / 2$. Therefore we can find $t$ of the solutions using the procedure of \cref{prop:most-sol-grover} with probability $\geq 1-\rho/3$, using
\begin{equation} \label{query-cost-random-subset}
  \bigO{\sqrt{\frac{N}{k}} (t +\log(1/\rho))} =  \bigO{\sqrt{\frac{N}{k}} \left(\frac{k}{\lambda} + \log(1/\rho) \right)}
\end{equation}
queries and
\begin{equation} \label{gate-cost-random-subset}
    \bigO{\sqrt{\frac{N}{k}} \left(\frac{k}{\lambda} + \log(1/\rho) \right) \log(N)}
\end{equation}
gates.
We remark here that these upper bounds hold because $t \leq k/2<k$. Indeed, under that assumption on $t$ and $k$ we have $k(H_{k}-H_{k-t}) \leq 2(t+1)$ by \cref{lem:harmonic numbers}, and moreover the factor $1/\ln(3k/(k+2(t-1)))$ is $\Theta(1)$ (it lies between $1/\ln(3)$ and $1/\ln(3/2)$). This shows that calling \GroverCoupon{} with $R = 6 \ln(2) (t+1) + 2 \ln(1/\rho) \ln(3/2) \in \Theta(t+\log(1/\rho))$ has the desired behaviour.

Let $a_1<a_2<\dots<a_{t}$ denote the found indices for which $x_{a_{j}} = 1$ and define the intervals $I_0 = \{1,\ldots,a_1-1\}$, $I_t = \{a_t+1,\ldots,N\}$, and, for $j \in [t-1]$, $I_j = \{a_j+1,\ldots,a_{j+1}-1\}$. We use
$k_j$ to denote the (unknown) number of marked elements in $I_j$, so in particular $\sum_{j=0}^t k_j \leq k-t$.
Then by \cref{lem:prob-of-run}, the probability that there is a $k_j$ larger than $\ell := \frac{k}{t} (\log(k)+\log(3/\rho))$ is at most
\begin{align*}
    (k-\ell+1)\left(1-\frac{t}{k}\right)^\ell
    \leq 2^{\log(k)} \left(1-\frac{t}{k}\right)^{ \frac{k}{t} (\log(k) + \log(3/\rho))}
    \leq 2^{\log(k)} \left(\frac{1}{2}\right)^{\log(k) + \log(3/\rho)} 
    = \rho/3.
\end{align*}
Here we used that $\ell \geq 1$, $(1 - \frac{t}{k})^{k/t} \leq \frac{1}{e} \leq \frac{1}{2}$, and~$\log(k) + \log(3/\rho) \geq 0$.\footnote{Note also that~$\ell \leq k$ because~$t \geq \log(6 \kest / \rho) \geq \log(3 k /\rho)$ by assumption; if~$\ell > k$, then the probability of having an interval of length~$\geq \ell$ is of course~$0$, and in this regime one may just as well run \GroverCertaintyMultiple{} on the whole string (and have zero failure probability).} 
For the rest of the argument we may thus assume that there is no interval with more than $\ell$ not-yet-found marked elements.

In the next step of our algorithm we search for all marked elements in each interval. To do so for the $j$th interval, we search over the elements from $[2^{\lceil\log(\abs{I_j})\rceil}]$ marking an element $i \in [2^{\lceil\log(\abs{I_j})\rceil}]$ if $x_{i+a_j} = 1$ and $i\leq \abs{I_j}$ (letting $a_0 = 0$).
One can implement this unitary using $\bigO{1}$ quantum queries and $\bigO{\log(N)}$ gates (to implement the addition and comparison).
For each interval, we first compute an estimate~$(k_j)_{\mathrm{est}}$ of~$k_j$ that satisfies $k_j /2 \leq (k_j)_{\mathrm{est}} \leq 3k_j/2$ using \cref{cor:factor3-approx}, with success probability~$\geq 1 - \rho/(3(t+1))$.
The associated query cost is~$\bigO{\sqrt{\abs{I_j}/(k_j + 1)} \log(t/\rho)}$, and it uses~$\bigO{\sqrt{\abs{I_j}/(k_j + 1)} \log(t/\rho) \log(N)}$ additional gates.
Then \Cref{lem:all-sol-grover-slow} shows that we can find all marked elements in the~$j$-th interval with probability~$1$ using
$\bigO{\sqrt{\abs{I_j} \, (k_j)_{\mathrm{est}}}}$ quantum queries and $\bigO{\sqrt{\abs{I_j}} \, (k_j)_{\mathrm{est}}^{3/2} \log(N)}$ additional gates.
By a union bound, with probability~$\geq 1 - \rho/3$, all~$(k_j)_{\mathrm{est}}$ are correct, and this step has a total query complexity of
\begin{equation} \label{query-cost-all-intervals}
\bigO{\sum_{j=0}^{t} \sqrt{\abs{I_j}k_j}+ \sum_{j=0}^{t}\sqrt{\abs{I_j}}\log(t/\rho)} = \bigO{\sqrt{Nk}+\sqrt{Nt}\log(t/\rho)} = \bigO{\sqrt{Nk}(1+\frac{\log(k/\rho\lambda)}{\sqrt{\lambda}})},
\end{equation}
where the first step uses Cauchy--Schwarz for both terms (reading~$\sqrt{\abs{I_j}}$ as~$\sqrt{\abs{I_j} \cdot 1}$ for the second term) and~$\sum_{j=0}^t \abs{I_j} \leq N$, $\sum_{j=0}^t k_j = k$.
To analyze the gate complexity of this step, we first bound $\sum_{j=0}^t k_j^3$. We have $\norm{\vec{k}^2}_\infty\leq \ell^2 = \bigO{\lambda^2 \log^2({3}k/\rho)}$ where $\vec{k}$ is the vector with entries $k_j$ and $\vec{k}^2$ is the entrywise square of $\vec k$. As we also have $\norm{\vec{k}}_1 \leq k$ we get $\sum_{j=0}^t k_j^3 = \langle \vec{k} , \vec{k}^2\rangle \leq \norm{\vec{k}}_1\norm{\vec{k}^2}_\infty = \bigO{k \ell^2}$. 
Then the gate complexity of the final search steps becomes:
\begin{align} \notag
&\bigO{\left(\sum_{j=0}^t\sqrt{\abs{I_j}k_j^{3}} +\sum_{j=0}^t\sqrt{\abs{I_j}} \log(t/\rho )\right)\log(N) }\\
       &=\bigO{\sqrt{N}\left(\sqrt{\sum_{j=0}^t k_j^{3}} + \sqrt{t} \log(t/\rho )\right)\log(N) } \notag \\
        &=\bigO{\sqrt{N}\left(\sqrt{k}\ell + 1 + \sqrt{t} \log(t/\rho )\right)\log(N) } \notag\\
        &=\bigO{\sqrt{N}\left(\sqrt{k} \lambda \log(k/\rho) + \sqrt{k/\lambda} \log(k/\rho \lambda)\right)\log(N) } \notag\\
        &=\bigO{\sqrt{Nk}\left(\lambda \log(k/\rho) + \sqrt{1/\lambda} \log(k/\rho \lambda)\right)\log(N) } \notag\\
        &=\bigO{\sqrt{Nk}\lambda \log(k/\rho)\log(N)}, \label{gate-cost-all-intervals}
\end{align}
where we again used Cauchy--Schwarz in the first step, and the total error probability is bounded by $\rho/3+\rho/3+(t+1)\cdot \frac{\rho}{3(t+1)} = \rho$.

To conclude, the upper bound on the total query complexity follows by combining \cref{query-cost-random-subset,query-cost-all-intervals}:
\begin{align*}
    &O\left(\vphantom{\sqrt{\frac{N}{k}}}\right.\underbrace{\sqrt{\frac{N}{k}} \left(\frac{k}{\lambda} + \log(1/\rho)\right) }_{\text{sample $t$ elements}} + \underbrace{\sqrt{Nk}\left(1+\frac{1}{\sqrt{\lambda}}\log(k/\rho\lambda)\right)}_{\text{find remaining elements}}\left.\vphantom{\sqrt{\frac{N}{k}}}\right) \\
    &= \bigO{\sqrt{Nk} \left(1 + \frac{1}{\sqrt{\lambda}} \log(k/\rho \lambda) \right) + \sqrt{\frac{N}{k}} \log(1/\rho)} = \bigO{\sqrt{Nk} \left(1 + \frac{1}{\sqrt{\lambda}} \log(k/\rho \lambda) \right)}.
\end{align*}
Here the first equality uses that $\sqrt{\frac{N}{k}} \frac{k}\lambda \leq \sqrt{Nk}$ since $\lambda \geq 1$. The second equality follows since $\log(1/\rho) \leq \log(6\kest/\rho)$ and, by assumption, $\log(6\kest/\rho) \leq \lceil \kest/\lambda\rceil = t \leq k$.
A similar argument using \cref{gate-cost-random-subset,gate-cost-all-intervals} and $\lambda \geq 1$, establishes the desired gate complexity:
\begin{align*}
    & O \parens[\Bigg]{\underbrace{\sqrt{\frac{N}{k}} \left(\frac{k}{\lambda} + \log(1/\rho) \right) \log(N)}_{\text{sample $t$ elements}} + \underbrace{\sqrt{Nk}\lambda \log(k/\rho)\log(N)}_{\text{find remaining elements}}}\\
    & = \bigO{\sqrt{Nk}\left(\frac{1}{\lambda} + \lambda\log(k/\rho)\right)\log(N) + \sqrt{\frac{N}{k}}\log(1/\rho)\log(N)}\\
    & = \bigO{\sqrt{Nk}\lambda\log(k/\rho)\log(N)}. \qedhere 
    \end{align*}
\end{proof}

\section{Improved query complexity for approximate summation}
\label{sec:fastsum}
In this section, we provide an algorithm~\ApproxSum{}, which given quantum query access to a binary description of $v \in [0,1]^N$, in the sense of~\cref{dfn:vector oracle access}, finds a multiplicative $\delta$-approximation of $s = \sum_{i=1}^N v_i$ with probability $\geq 1 - \rho$ using
\begin{equation}
    \label{eq:informal summing complexity}
    \bigO{\sqrt{\frac{N}{\delta} \log(1/\rho)}}
\end{equation}
quantum queries and a similar gate complexity (with only a polylogarithmic overhead).
In the above~\eqref{eq:informal summing complexity} we have made very mild assumptions on the value of~$\rho$ and $\delta$; a precise statement is given in~\cref{thm:fast-summing}.
The algorithm is given in~\ApproxSum{}.
By slightly perturbing the entries of~$v$, we may assume without loss of generality that all entries of~$v$ are distinct; we shall make this assumption throughout this section, and have made this assumption in the description of the algorithm as well.

We briefly explain the overall strategy.
Recall from the proof of~\cref{thm:approx counting} that it is useful to preprocess the vector~$v$ by using quantum maximum finding to find~$\vmax = \max_{i \in [N]} v_i$, and then to use amplitude estimation on the vector~$w = v / \vmax$.
We take this approach slightly further: we first find the largest~$k$ entries $z_1,\ldots,z_k$ of $v$, where~$k = \Theta(p N)$ for~$p \in (0,1)$, and sum their values classically.
Let~$\tilde z$ be the smallest value among the~$z_1, \dotsc, z_k$.\footnote{We actually first compute a good value of~$\tilde z$ using a quantile estimation subroutine~\cite[Thm.~3.4]{hamoudi:QuantumSubGaussianMean2021} and then find all the~$z_j$'s. Alternatively, one could use~\cite[Thm.~3.4]{durrQuantumQueryComplexity2006} to find all~$\Theta(pN)$ largest elements directly, but our approach has the advantage of being able to use the better~$\rho$-dependence of our version of Grover search.}
For the next part, we treat the corresponding entries of~$v$ as zero: checking whether one exceeds the threshold~$\tilde z$ is a binary comparison, hence can be done in superposition without explicitly using their indices, and so with one query to~$v$ we can implement quantum oracle access to the vector~$w \in [0,1]^N$ defined by
\begin{equation*}
    w_i = \begin{cases}
    \frac{v_i}{\tilde z}  & \text{ if } v_{i} < \tilde z\\
    0 & \text{ else.}
    \end{cases}
\end{equation*}
This has the effect of amplifying the small elements in~$v$ at no extra cost. We then use amplitude estimation to compute~$\sum_{i=1}^N w_i$ with additive precision~$\bigO{\delta s /\tilde z}$ (without knowing $s$). This yields an additive $\delta s$-approximation of~$\sum_{i=1}^N v_i$ (i.e.,~a multiplicative $\delta$-approximation), where we use that
\begin{equation*}
    \sum_{i=1}^N v_i = \sum_{i=1}^k z_i + \tilde z \sum_{i=1}^N w_i
\end{equation*}
To balance the costs of these two stages we need to carefully choose $\tilde z$. We do so by estimating the $p$-th quantile of the vector.
We first give an algorithm \ApproxSum{} whose complexity depends on the quantile~$p$ and then give a suitable choice for $p$ that allows us to obtain \eqref{eq:informal summing complexity}, see \cref{thm:fast-summing} and \cref{cor:summing simple corollary}.

\begin{procedure}[t]
  \caption{ApproxSum($U_v$, $\delta$, $p$, $\lambda$, $\rho$)}\label{alg:approx sum}
  \Input{Quantum query access $U_v$ to~$(0,b)$-fixed point representations of~$v \in [0,1]^N$, $\delta \in (0,1)$, $p \in (0,1)$, $\lambda \geq 6$, failure probability $\rho > 0$.}
  \Output{A real number $\tilde{s}$.}
  \Guarantee{With probability $\geq 1 - \rho$, $\tilde{s}$ is a multiplicative $\delta$-approximation of $s$.}
  \Analysis{\cref{thm:fast-summing}}
  \medskip

  use~\cref{thm:quantile estimation} to compute~$\tilde z \in [0,1]$ such that with probability~$\geq 1 - \rho/4$, $Q(p) \leq \tilde z \leq Q(cp)$, where~$c < 1$ is a universal constant and~$Q$ is defined in~\cref{eq:quantile definition}\;
  let~$x \in \{0,1\}^N$ be defined by~$x_i = 1$ if~$v_i \geq \tilde z$ and~$x_i = 0$ otherwise\;
  let~$U_x$ implement quantum query access to~$x$ by applying~$U_v$, comparing to~$\tilde z$, and uncomputing~$U_v$\;
  compute estimate $\kest$ of $k = \abs{x}$ satisfying~$\frac{k}{2} \leq \kest \leq \frac{3k}{2}$ with probability~$\geq 1-\rho/4$ using \cref{cor:factor3-approx}\;
  use~\GroverMultipleFast{}$(U_x, \kest, \rho/4,\lambda)$ to find all indices $i_1, \dotsc, i_k$ such that $x_{i_j} = 1$\;

  \smallskip
  \uIf{$\tilde z = 0$}{
    \Return{$\sum_{j=1}^k v_{i_j}$}\;
  } \Else{

  construct unitary~$U_w$ for query access to $w \in [0, 1]^N$ where $w_i = 0$ if $v_i \geq \tilde z$ and $w_i = v_i / \tilde z$ otherwise\;
  let $U$ be a unitary such that~$U \ket{0} = \ket{\psi}$ given by
  \vspace{-3mm}
  \begin{equation*}
    \vspace{-3mm}
      \ket{\psi} = \frac{1}{\sqrt{N}} \sum_i \ket{i} (\sqrt{\tilde{w}_i} \ket{1} + \sqrt{1 - \tilde{w}_i}\ket{0})
  \end{equation*}
  where $\alpha_i$ is a $\lceil\log(4 N / \delta)\rceil$-bit approximation of $\arcsin(\sqrt{w_i})$, and $\sqrt{\tilde w_i} = \sin(\alpha_i)$\;
  use \AmpEst{}$(U,M)$ with~$M = \lceil 12 \pi \sqrt{\delta^2 p c} \rceil$, increased to the next power of~$2$ if necessary, with~$c < 1$ from~\cref{thm:quantile estimation}, to compute $\tilde a \approx \sum_i \tilde{w}_i / N$,   and repeat~$\bigO{\log(1/\rho)}$ times and take the mean of the outputs to achieve success probability~$\geq 1-\rho/4$\;

  \Return{$\sum_{j=1}^k v_{i_j} + N \tilde z \tilde a$}\;
  }
\end{procedure}

\medskip

We use the following lemma to derive a bound on the required precision for certain arithmetic operations.
\begin{lem}[{\cite[Lem.~7]{bhmt:countingj}}]
  \label{lem:sin stability}
  If $a = \sin^2(\theta_a)$ and $\tilde a = \sin^2(\tilde \theta_a)$ for $\theta_a, \tilde \theta_a \in [0, 2\pi]$, then $\abs{\tilde \theta_a - \theta_a} \leq \delta$ implies $\abs{\tilde a - a} \leq 2 \delta \sqrt{a (1-a)} + \delta^2$.
\end{lem}

For the quantile estimation, we use a subroutine from~\cite{hamoudi:QuantumSubGaussianMean2021}.
Let~$v \in [0,1]^N$.
Then for~$p \in (0,1)$, we define the~$p$-quantile~$Q(p) \in [0,1]$ by
\begin{equation}
    \label{eq:quantile definition}
    Q(p) = \sup \{z \in [0,1] : \abs{\{i \in [N] : v_i \geq z\}} \geq p N\}.
\end{equation}
In words, $Q(p)$ is the largest value~$z \in [0,1]$ such that there are at least~$p N$ entries of~$v$ which are larger than~$z$.
The subroutine we invoke allows one to produce an estimate for~$Q(p)$, in the following sense:
\begin{thm}[{\cite[Thm.~3.4]{hamoudi:QuantumSubGaussianMean2021}}]
  \label{thm:quantile estimation}
  There exists a universal constant~$c \in (0,1)$ such that the following holds:
  Let~$v \in [0,1]^N$ and let~$U_v$ be a unitary implementing quantum oracle access to~$v$.
  Then~$\bigO{\log(1/\rho) / \sqrt{p}}$ applications of~controlled-$U_v$ and controlled-$U_v^\dagger$ suffice to find, with probability~$\geq 1-\rho$, a value~$\tilde z$ such that~$Q(p) \leq \tilde z \leq Q(c p)$.
  The algorithm uses an additional~$\bigO{(\log(1/\rho)/\sqrt{p})\, b \log(b) \log(N)}$ gates.
\end{thm}
The actual access model for which the above theorem holds is more general, but we have instantiated it for our setting.
The gate complexity overhead follows from having to implement their access model from ours, which involves arithmetic and comparisons on the fixed point representations we use, and the fact that the underlying technique is amplitude amplification.
We now get to the main theorem of this section, which proves the correctness of \ApproxSum{} and analyzes its complexity.

\begin{thm}
  \label{thm:fast-summing}
  Let~$v \in [0,1]^N$, let~$U_v$ be a unitary implementing quantum query access to~$(0, b)$-fixed point representations of~$v$, and let~$\delta \in (0, 1)$.
  Let~$p, \rho \in (0,1)$ and choose $6 \leq \lambda \leq \min\{c p N / \log(p N/ \rho), \log(c p N/\rho)^2\}$.
  Then \ApproxSum{} computes, with probability~$\geq 1 - \rho$, a~multiplicative $\delta$-approximation of~$s = \sum_{i=1}^N v_i$.
  It uses
  \vspace{-1mm}
  \begin{equation*}
  \vspace{-1mm}
    \bigO{\frac{\log(1/\rho)}{\sqrt{p}} + \sqrt{\frac{N}{Np + 1}} \log(1/\rho) + N \sqrt{p} \left( 1 + \frac{1}{\sqrt{\lambda}} \log(N p / \lambda \rho))\right) + \frac{1}{\delta \sqrt{p}} \log(1/\rho)}
  \end{equation*}
  quantum queries, and the number of additional gates is bounded by
  \vspace{-1mm}
  \begin{align*}
  \vspace{-2mm}
    &O\left(\frac{\log(1/\rho)}{\sqrt{p}} \, b \log(b) \log(N) + \sqrt{\frac{N}{Np + 1}} \log(1/\rho) \log(N) \right. \\[-0.5ex]
    &\qquad \left.\vphantom{\sqrt{\frac{N}{Np + 1}}}+ N \sqrt{p}\lambda\log(p N/\rho)\log(N) + \frac{1}{\delta \sqrt{p}} b \log(b) \log(N / \delta) \log^2 \log(N/\delta) \log(1/\rho)\right).
  \end{align*}
\end{thm}
Before we give the proof, we discuss two useful regimes for~$p$ and~$\lambda$:
\begin{cor}
  \label{cor:summing simple corollary}
  Let~$v \in [0,1]^N$, let~$U_v$ be a unitary implementing quantum oracle access to~$(0, b)$-fixed point representations of~$v$, and let~$\delta \in (0,1)$.
  Then we can find, with probability~$\geq 1 - \rho$, a~multiplicative $\delta$-approximation of~$s = \sum_{i=1}^N v_i$, using:
  \begin{itemize}
      \item $\bigO{\sqrt{N \log(1/\rho) / \delta}}$ quantum queries, when~$p = \Theta(\log(1/\rho) / (\delta N)) < 1$ and we choose $\lambda = \min\{ c p N / \log(6 p N / \rho)), \log(c p N / \rho)^2 \} \geq 6$, and using~$\sqrt{N/\delta} \poly(\log(1/\rho), b, \log(N), \log(1/\delta))$ additional gates, or 
      \item $\bigO{\sqrt{N/\delta} \log(1/\rho)}$ quantum queries when~$p = \Theta(1 / (\delta N)) < 1$ and we choose~$\lambda = 6$, and using~$\sqrt{N/\delta} \, \poly(\log(1/\rho), b, \log(N), \log(1/\delta))$ additional gates.
  \end{itemize}
\end{cor}
\begin{proof}[Proof of \cref{thm:fast-summing}]
  We assume without loss of generality that all the entries of~$v$ are distinct.
  If this is not the case, one can perturb the~$i$-th entry of~$v$ by~$i 2^{-\ell}$ for some sufficiently large~$\ell = \Omega(\log(N) + b)$, where we recall that~$b$ is the number of bits describing~$v_i$, and discarding these trailing bits from the output value~$\tilde s$.

  We use~\cref{thm:quantile estimation} to find a value~$\tilde z$ such that the number of elements of~$v$ that are at least as large as~$\tilde z$, is at most~$p N$ and at least~$c p N$.
  The number of quantum queries is
  \begin{equation*}
  \label{eq:quantile query cost}
    \bigO{\frac{\log(1/\rho)}{\sqrt{p}}},
  \end{equation*}
  and the number of additional gates used is
  \begin{equation*}
  \label{eq:quantile gate cost}
   \bigO{\frac{\log(1/\rho)}{\sqrt{p}} \, b \log(b) \log(N)}.
  \end{equation*}

  Let~$k = \abs{\{i \in [N] : v_i \geq \tilde z \}}$.
  By the assumption that the~$v_i$ are all distinct, $cpN \leq k \leq pN$.
  We next compute a multiplicative $\frac12$-approximation of~$k$ using \cref{cor:factor3-approx}.
  This uses
  \begin{equation*}
  \label{eq:count big numbers query cost}
    \bigO{\sqrt{N/(k+1)} \log(1/\rho)}
  \end{equation*}
  quantum queries and
  \begin{equation*}
  \label{eq:count big numbers gate cost}
    \bigO{\sqrt{N/(k+1)} \log(1/\rho) \log(N)}
  \end{equation*}
  additional gates.
  The next step is to find all~$k$ such elements using~\GroverMultipleFast{} (\cref{thm:all-sol-grover-fast}).
  This uses
  \begin{equation*}
  \label{eq:find big numbers query cost}
    \bigO{\sqrt{N k} \left( 1 + \frac{1}{\sqrt{\lambda}} \log(k / (\lambda \rho))\right)}
  \end{equation*}
  quantum queries and
  \begin{equation*}
  \label{eq:find big numbers gate cost}
    \bigO{\sqrt{Nk}\lambda\log(k/\rho)\log(N)}
  \end{equation*}
  additional gates.

  Let~$z_1, \dotsc, z_k$ be the entries of~$v$ that are $\geq \tilde z$.
  Then
  \begin{equation*}
    \sum_{i=1}^N v_i = \sum_{j=1}^k z_j + \tilde z \sum_{i=1}^N w_i
  \end{equation*}
  where
  \begin{equation*}
    w_i = \begin{cases}
      \frac{v_i}{\tilde z} & \text{if } v_i < \tilde z \\
      0 & \text{otherwise.}
    \end{cases}
  \end{equation*}
  As we have found all the~$z_j$'s, we can compute their sum exactly; therefore, to determine a~multiplicative $\delta$-approximation of~$s$, we must produce an additive~$\delta s$-approximation of~$\tilde z \sum_{i=1}^N w_i$.
  Let~$\eps := \delta s$; note that we do not know~$s$ as we do not know~$\delta$.
  Then we have to approximate~$\frac{1}{N} \sum_{i=1}^N w_i$ with precision~$\eps / (N \tilde z)$.
  For this, we use amplitude estimation as follows.
  First, one can implement query access to~$U_w$ by using two quantum queries to $v$ and $\bigO{b\log(b)}$ non-query gates, by querying an entry, comparing the entry to~$\tilde z$, and conditional on the comparison uncomputing the query, and lastly performing the division by~$\tilde z$.
  From this, we can construct a unitary $U$ with~$U \ket{0} = \ket{\psi}$ satisfying
  \begin{equation*}
    \ket{\psi} = \frac{1}{\sqrt{N}} \sum_i \ket{i} \left(\sqrt{\tilde{w}_i}\ket{1}+\sqrt{1-\tilde{w}_i}\ket{0}\right),
  \end{equation*}
  where~$\tilde{w}_i$ is close to~$w_i$.
  One can implement such a unitary as follows.
  First, set up a uniform superposition over the index register using $\bigO{\log(N)}$ gates.
  Use~$U_w$ to load binary descriptions of the entries of~$w$.
  Calculate a $\lceil\log_2(4 N / \delta)\rceil$-bit approximation $\alpha_i$ of $\arcsin(\sqrt{w_i})$ using $\bigO{\log(b N/\delta) \log^2\log(b N/\delta)}$ gates~\cite[Ch.~4]{brentzimmermann2010}.
  Then conditionally rotate the last qubit from $0$ to $1$ over angles $\pi/4$, $\pi/8$, et cetera, depending on the bits of $\alpha_i$.
  Lastly, we uncompute~$\alpha_i$ and~$U_w$ to return work registers to the zero state, and we have obtained the desired state~$\ket{\psi}$, where~$\sqrt{\tilde{w}_i} = \sin(\alpha_i)$.
  We now show that~$\tilde w_i = \sin(\alpha_i)^2$ is close to~$w_i$, and hence
  \begin{equation*}
    a := \frac1N \sum_{i=1}^N \tilde w_i = \norm{\ket{\psi_1}}^2
  \end{equation*}
  is close to~$\frac1N \sum_{i=1}^N w_i$.
  \Cref{lem:sin stability} shows that if $\abs{\alpha_i - \arcsin(\sqrt{w_i})} \leq \xi$, then
  \begin{equation*}
    \abs{\tilde w_i - w_i} = \abs{\sin^2(\alpha_i) - w_i} \leq 2 \xi \sqrt{w_i (1 - w_i)} + \xi^2 \leq \xi + \xi^2.
  \end{equation*}
  Since~$\alpha_i$ is a~$\lceil\log_2(4 N / \delta)\rceil$-bit approximation of~$\arcsin(\sqrt{w_i})$, we may apply the above with~$\xi = \delta/(4 N)$ for every~$i \in [N]$.
  Because~$s \geq \tilde z$, $\delta = \eps/s \leq \eps / \tilde z$, and~$\delta \leq 1$, so the total error satisfies
  \begin{equation*}
    \abs{a - \frac1N \sum_{i=1}^N w_i} \leq \frac1N \sum_{i=1}^N \abs{\tilde w_i - w_i} = \xi  + \xi^2 \leq \frac{\delta}{4 N} + \frac{\delta^2}{16 N^2} \leq \frac{\eps}{2 N \tilde z}.
  \end{equation*}
  Next, we use this to derive an upper bound on~$a$:
  \begin{equation*}
    a \leq \frac{\eps}{2 N \tilde z} + \frac1N \sum_{i=1}^N w_i = \frac{\eps}{2 N \tilde z} + \frac{1}{N} \sum_{i : v_i < \tilde z} \frac{v_i}{\tilde z} \leq \frac{2 s}{N \tilde z},
  \end{equation*}
  where the last inequality uses~$\eps = \delta s \leq s$ and $\sum_{i:v_i <\tilde z} v_i \leq s$.
  Therefore, using~\AmpEst{} with~$M$ applications of~$U$ yields a number~$\tilde a \in [0,1]$ with
  \begin{equation*}
    \abs{\tilde a - a} \leq 2 \pi \frac{\sqrt{a (1-a)}}{M} + \frac{\pi^2}{M^2} \leq 2 \pi \frac{\sqrt{2 s / (N \tilde z)}}{M} + \frac{\pi^2}{M^2}
  \end{equation*}
  by~\cref{lem:ampest}.
  We now determine an appropriate number of rounds~$M$ to be used for amplitude estimation.
  We will choose~$M$ such that $\abs{\tilde a - a} \leq \frac12 \eps / (N \tilde z)$; if we do so, then by the triangle inequality ~$\abs{\tilde a - \frac1N \sum_{i=1}^N w_i} \leq \eps / (N \tilde z)$.
  The claim is that any~$M \geq 12 \pi \sqrt{N \tilde z / (\eps \delta)}$ suffices, as then
  \begin{equation*}
    2 \pi \frac{\sqrt{2 s/(\tilde z N)}}{M} \leq \frac{2 \pi \sqrt{2}}{12 \pi} \frac{\sqrt{s / (N \tilde z)}}{\sqrt{N \tilde z / (\eps \delta)}} = \frac{\sqrt{2}}{6} \frac{\eps}{N \tilde z} \leq \frac{1}{4} \frac{\eps}{N \tilde z},
  \end{equation*}
  and, using~$\delta \leq 1$, 
  \begin{equation*}
    \frac{\pi^2}{M^2} \leq \frac{\eps \delta}{144 N \tilde z} \leq \frac{1}{4} \frac{\eps}{N \tilde z}.
  \end{equation*}

  Even though we do not know~$\eps$, by choosing~$p$ carefully, we can enforce upper bounds on~$\tilde z$ and give a safe choice for~$M$.
  We use that the number of entries~$k$ which are at least~$\tilde z$ satisfies~$k \geq c p N$, so that
  \begin{equation*}
    c p \, N \, \tilde z \leq \sum_{i : v_j \geq \tilde z} v_j \leq s,
  \end{equation*}
  i.e.,~$\tilde z \leq s / (c p N)$.
  Therefore it suffices to take~$M = 12 \pi / \sqrt{\delta^2 p c}$, as this satisfies
  \begin{equation*}
    M = 12 \pi \sqrt{\frac{1}{\delta^2 p c}} = 12 \pi \sqrt{\frac{s}{\delta \eps p c}} \geq 12 \pi  \sqrt{\frac{N \tilde z}{\delta \eps}},
  \end{equation*}
  This guarantees that~$\abs{\tilde a - \frac1N \sum_{i=1}^N w_i} \leq \eps/(N \tilde z)$, and the output value~$\tilde s = \sum_{j=1}^k z_j + N \tilde z \tilde a$ satisfies
  \begin{equation*}
    \abs{\tilde s - s} \leq \eps = \delta s.
  \end{equation*}
  The number of quantum queries used for this step is therefore~$\bigO{M} = \bigO{1/(\delta \sqrt{p})}$, and the number of additional gates used is~$\bigO{M b \log(b) \log(N / \delta) \log^2 \log(N/\delta)}$.
  To amplify the success probability to~$1 - \rho$, we repeat the above procedure~$\log(1/\rho)$ many times and output the median of the individual estimates.
  The query- and gate complexity of the entire algorithm follow by combining those of the four parts: the quantile estimation, the approximate counting, Grover search for finding all large elements, and amplitude estimation for approximating the sum of the small elements.
\end{proof}

\section*{Acknowledgements}
We would like to thank Ronald de Wolf for helpful discussions and comments on an early version of this work, and Yassine Hamoudi for helpful discussion regarding \cite{hamoudi:QuantumSubGaussianMean2021}.
We also thank anonymous referees for their feedback.
HN acknowledges support by the Dutch Research Council (NWO grant OCENW.KLEIN.267), by the European Research Council~(ERC) through ERC Starting Grant 101040907-SYMOPTIC and ERC Grant Agreement No. 81876432, and by VILLUM FONDEN via the QMATH Centre of Excellence (Grant No. 10059).
JvA was supported by the Dutch Research Council (NWO/OCW), as part of QSC (024.003.037) and by QuantumDelta~NL.

\bibliographystyle{alphaurl}
\bibliography{references}

\appendix

\section{Classical lower bound for approximate counting} \label{sec:classicalLB}
Let~$x \in \{0,1\}^N$ be a bit string.
We show in this appendix the (well-known) result that any randomized classical algorithm that computes with probability~$\geq 5/6$ a $\delta$-multiplicative approximation of the Hamming weight~$k = \abs{x}$, must make at least~$\Omega(\min(N,\frac{N}{\delta^2 k}))$ queries to~$k$.
Although a matching upper bound on the sample complexity follows from an application of Chebyshev's inequality, and the lower bound is claimed in~\cite{aaronson&rall:qcounting}, we did not succeed in locating a proof for non-constant~$\delta$ in the literature. We therefore include a proof of the following statement. 
\begin{thm}
  \label{thm:classical randomized counting lower bound}
  There exists a universal constant~$C > 0$ such that the following holds.
  Let~$\delta > 0$ be such that $k \delta$ is an integer and $k (1 + \delta) \leq N$.
  Suppose $\mathcal A$ is a randomized query algorithm such that for all~$x \in \{0,1\}^N$ with~$\abs{x} \in \{k, k (1 + \delta)\}$,
  \begin{equation*}
    \Pr_{r \sim \Unif(\{0,1\}^R)}[\mathcal A(x, r) = \abs{x}] \geq \frac56,
  \end{equation*}
  where $R$ is the number of used random bits, and $\Unif(\{0,1\}^R)$ refers to the uniform distribution on them.
  Assume that~$\mathcal A$ makes $t \geq 0$ queries, independent of the input~$x$, or the randomness~$r$ used by the algorithm.
  Then
  \begin{equation*}
    t \geq C \min\{ N, \frac{N}{\delta^2 \abs{x}}\}
  \end{equation*}
\end{thm}

At a high level, our proof boils down to showing that if a $t$-query algorithm $\mathcal A$ succeeds with high probability, then the total variation distance between $\Hyp(N, k, t)$ and $\Hyp(N, k(1+\delta), t)$ must be $\Omega(1)$. Here $\Hyp(N, \ell, t)$ is the distribution on the number of observed marked elements if one draws $t$ elements from a set of size $N$ of which $\ell$ elements are marked, without replacement. The lower bound on $t$ then follows from a $t$-dependent upper bound on this total variation distance, for which we state the necessary ingredients from the literature first. We let~$\Bin(t, p)$ the binomial distribution with parameters~$t \geq 1$ and~$p \in (0,1)$, corresponding to~$t$ independent Bernoulli trials, each of which succeeds with probability~$p$. 
First, we state the following bound~\cite[Thm.~3]{diaconisFiniteExchangeableSequences1980} which shows that when the number of samples~$t$ is small compared to~$N$, sampling with and without replacement yield approximately the same distribution.
\begin{thm}
  \label{thm:distance hypergeometric binomial}
  The total variation distance between~$\Hyp(N, \ell, t)$ and~$\Bin(t, \ell/N)$ is at most~$4 t/N$.
\end{thm}

Next, we use the following estimate on the total variation distance between two binomial distributions with the same~$t$, but distinct success probability~\cite{roosBinomialApproximationPoisson2001,adellExactKolmogorovTotal2006}:
\begin{thm}
    \label{thm:distance binomial with different parameters}
    Let~$t \geq 1$, $p \in (0,1)$ and~$\delta \in (0, 1)$ such that~$p (1 + \delta) < 1$.
    The total variation distance between~$\Bin(t, p)$ and~$\Bin(t, p (1 + \delta))$ is at most
    \begin{equation*}
        \frac{\sqrt{e}}{2} \frac{\tau}{(1 - \tau)^2}, 
    \end{equation*}
    assuming~$\tau < 1$, where
    \begin{equation*}
        \tau = \delta \sqrt{\frac{p (t+2)}{2 (1-p)}}.
    \end{equation*}
\end{thm}
\begin{proof}
    This follows from~\cite[eq.~(15)]{roosBinomialApproximationPoisson2001}.
    We apply their bound on the distance to~$\Bin(t, p)$ with the following parameters: $s = 0$, 
    $n$ is our $t$, the random variable~$S_n$ is the sum of~$t$ Bernoulli random variables with success probability~$p + x$ with~$x = \delta p$, hence its distribution~$P^{S_n}$ is~$\Bin(t, p+x)$, and
    \begin{equation*}
        \gamma_1(p) = t x, \, \gamma_2(p) = t x^2, \, \eta(p) = 2 t x^2 + t^2 x^2, \, \theta(p) = \frac{2 t x^2 + t^2 x^2}{2 t p (1-p)}.
    \end{equation*}
    The upper bound given is then $\dTV(P^{S_n}, \Bin(t, p)) \leq \frac{\sqrt{e}}{2} \frac{\sqrt{\theta(p)}}{(1 - \sqrt{\theta(p)})^2}$. The claimed bound in the theorem then follows from~$\theta(p) = \tau$, using~$x = p \delta$.
\end{proof}

Via the triangle inequality, the above two theorems suffice to upper bound the total variation distance between $\Hyp(N,k,t)$ and $\Hyp(N,k(1+\delta),t)$ (for a precise statement, see the proof below). We are now ready to prove \cref{thm:classical randomized counting lower bound}.
\begin{proof}[Proof of~\cref{thm:classical randomized counting lower bound}.]
  For an integer $\ell$, let $X_\ell \subseteq \{0,1\}^N$ be the set of bit strings with Hamming weight $\ell$. We use Yao's minimax principle to lower bound the number of queries required by a randomized algorithm 
 that outputs $|x|$ with probability $\geq 5/6$ on every input $x \in X_k \cup X_{k(1+\delta)}$. That is, we exhibit a distribution $\mathcal D$ on $X_k \cup X_{k(1+\delta)}$ for which every deterministic algorithm that computes $|x|$ on a $5/6$ fraction of the inputs, weighted according to $\mathcal D$, requires at least $C \min\{N, N/(\delta^2 k)\}$ queries for some universal constant $C$. Consider the distribution~$\mathcal D$ on inputs that with probability $1/2$ samples a uniformly random element from $X_k$, and with probability $1/2$ samples a uniformly random element from $X_{k (1+\delta)}$. Suppose $\mathcal A$ is a deterministic $t$-query algorithm that on input $x \sim \mathcal D$ correctly returns $\abs{x}$ with probability at least $5/6$ (where the probability is over the sample from $\mathcal D$). Note that we allow $\mathcal A$ to know $k$ and $\delta$.
  We show the desired lower bound on $t$.
  Let $A(x) = a$ denote the substring $(x_{i_1}, \dotsc, x_{i_t})$ of $x$ that corresponds to the $t$ queried indices $i_1, \dotsc, i_t$.
  Note that the output $\mathcal A(x) \in \{k, k(1+\delta)\}$ of the algorithm is deterministic and \emph{only} a function of $a$ (for a fixed algorithm).
  If one thinks of $\mathcal A$ as a decision tree, then the first index to be queried does not depend on $a$, and after every subsequent query, the next index to be queried is deterministic as a function of the previous queried indices and outcomes.
  It is also the case that the queried indices $i_1, \dotsc, i_t$ are a function of the query outcomes $a$!
  Therefore, we may view $\mathcal A(x)$ as a just a function of $a = A(x)$.
  Let $B \subset X_k \cup X_{k(1+\delta)}$ be the set of $a \in \{0,1\}^t$ on which $\mathcal A(a)$ outputs $k (1+\delta)$.

  Let $P_\ell$ be the distribution on $a = A(x) \in \{0,1\}^t$ induced by $x \sim \Unif(X_\ell)$, where the latter refers to the uniform distribution on~$X_\ell$.
  By assumption, $A$ can distinguish (with constant success probability) the distributions $P_k$ and $P_{k (1+\delta)}$.
  Therefore, the total variation distance between these distributions is large. Indeed, the probability, with respect to $\mathcal D$, that $\mathcal A$ outputs the wrong value of $\abs{x}$ is at most $1/6$, therefore $\mathcal A$ fails with probability at most $1/3$ when $x \sim \Unif(X_\ell)$ for both $\ell=k$ and $\ell=k(1+\delta)$, and hence
  \begin{align*}
    \frac13 + \frac13
    & \geq
    \Pr_{x \in_R  X_{k (1 + \delta)}}[\mathcal A(x) = k]
    +
    \Pr_{x \in_R  X_{k}}[\mathcal A(x) = k (1 + \delta)]
    \\
    &
    =
    \Pr_{a \sim P_{k (1 + \delta)}}[\mathcal A(a) = k]
    +
    \Pr_{a \sim P_{k}}[\mathcal A(a) = k (1 + \delta)]
    \\
    & 
    =
    1 + (P_k(B) - P_{k(1+\delta)}(B))
    \\
    &
    \geq
    1 - \dTV(P_k, P_{k (1 + \delta)}),
  \end{align*}
  so $\dTV(P_k, P_{k (1 + \delta)}) \geq 1/3$.
  
  We now relate~$P_\ell$ to the hypergeometric distribution $\Hyp(N, \ell, t)$, so that we can upper bound the above total variation distance as a function of~$t$.
  We prove that
  \begin{align*}
    P_\ell(a) = \Pr_{x \in X_\ell}[A(x) = a] = \frac{1}{\binom{t}{\abs{a}}}\Pr[W_\ell = \abs{a}]
  \end{align*}
  where~$W_\ell \sim \Hyp(N,\ell,t)$ is hypergeometrically distributed, i.e.,
  \begin{equation*}
    \Pr[W_\ell = \abs{a}] = \frac{\binom{\ell}{\abs{a}} \binom{N-\ell}{t-\abs{a}}}{\binom{N}{t}}.
  \end{equation*}
  We prove this by exploiting the permutation symmetry of the distribution on $X_\ell$, along with an iterative conditioning argument.
  Let $i_1, \dotsc, i_t$ denote the sequence of indices of $x$ queried, so that $A(x) = (x_{i_1}, \dotsc, x_{i_t})$.
  Recall that the $i_1, \dotsc, i_t$ may be chosen adaptively, but $i_{j+1}$ is determined completely from $i_1, \dotsc, i_j$ and $x_{i_1}, \dotsc, x_{i_j}$.
  Then
  \begin{align*}
    \Pr_{x \in X_\ell}[x_{i_1} = 1] = \frac{\ell}{N}.
  \end{align*}
  Moreover, one has
  \begin{align*}
    \Pr_{x \in X_\ell}[x_{i_{j+1}} = 1 | x_{i_1}, \dotsc, x_{i_j}] = \frac{\ell - \sum_{s=1}^j x_{i_s}}{N-j},
  \end{align*}
  because conditioned on the values of $x$ at the indices $i_1, \dotsc, i_j$, the distribution of $x$ becomes uniform among bit strings of Hamming weight $\ell - \sum_{s=1}^j x_{i_s}$ of length $N-j$ in the remaining position.
  Finally
  \begin{align*}
    \Pr_{x \in X_\ell}[A(x) = a] = \prod_{j=1}^t \Pr_{x \in X_\ell} [x_{i_j} = a_j | x_{i_1} =a_1, \dotsc, x_{i_{j-1}} =a_{j-1}]
  \end{align*}
  which after careful consideration is seen to be equal to
  \begin{align*}
    \frac{\ell \dotsb (\ell-\abs{a} + 1) \cdot (N-\ell) \dotsb (N-\ell-(t-\abs{a})+1)}{N \dotsb (N-t+1)} & = \frac{\ell!}{(\ell-\abs{a})!} \frac{(N-\ell)!}{(N-\ell-(t-\abs{a}))!} \frac{(N-t)!}{N!} \\
    & = \frac{\binom{N-t}{\ell-\abs{a}}}{\binom{N}{\ell}}
  \end{align*}
  Now because
  \begin{equation*}
    \Pr[W_\ell = \abs{a}] = \frac{\binom{\ell}{\abs{a}} \binom{N-\ell}{t-\abs{a}}}{\binom{N}{t}}
  \end{equation*}
  we see that we indeed have
  \begin{align*}
    \Pr_{x \in X_\ell}[A(x) = a] = \frac{1}{\binom{t}{\abs{a}}}\Pr[W_\ell = \abs{a}].
  \end{align*}
  Therefore we obtain
  \begin{align*}
    \dTV(P_k, P_{k(1+\delta)}) & = \frac12 \sum_{a \in \{0,1\}^t} \abs{P_k(a) - P_{k(1+\delta)}(a)} \\
    & = \frac12 \sum_{a \in \{0,1\}^t} \frac{1}{\binom{t}{\abs{a}}} \abs{\Pr[W_k = \abs{a}] - \Pr[W_{k(1+\delta)} = \abs{a}]} \\
    & = \frac12 \sum_{s = 0}^t \abs{\Pr[W_k = s] - \Pr[W_{k(1+\delta)} = s]} \\
    & = \dTV(\Hyp(N, k, t), \Hyp(N, k (1 + \delta), t)).
  \end{align*}

  We now give a~$t$-dependent upper bound on this total variation distance; combined with the assumption that $\dTV(P_k, P_{k(1+\delta)}) \geq 1/3$, this will lead to the right lower bound on~$t$.
  Let
  \begin{equation*}
    \tau
    = \frac{\delta k}{N} \sqrt{\frac{t + 2}{2 \frac{k}{N} (1 - \frac{k}{N})}}
    = \delta \sqrt{\frac{k (t + 2)}{2 (N-k)}}.
  \end{equation*}
  If $\tau \geq 1/2$, then
  \begin{equation*}
    t+2 \geq \frac{(N-k)}{2 \delta^2 k}.
  \end{equation*}
  and so in this case $t = \Omega(N / (\delta^2 k))$ (unless the latter is $O(1)$, in which case the query lower bound we aim for is constant and uninteresting).
  Otherwise, by the triangle inequality,
  \begin{align*}
    \dTV(\Hyp(N, k, t), \Hyp(N, k (1 + \delta), t)) & \leq \dTV(\Hyp(N, k, t), \Bin(t, k/N)) \\
    & + \dTV(\Bin(t, k/N), \Bin(t, k(1+\delta)/N)) \\
    & + \dTV(\Bin(t, k(1+\delta)/N), \Hyp(N, k (1 + \delta), t)) \\
    & \leq \frac{8 t}{N} + \frac{\sqrt{e}}{2} \frac{\tau}{(1 - \tau)^2} \\
    & \leq \frac{8 t}{N} + 2 \tau \sqrt{e},
  \end{align*}
  where we applied~\cref{thm:distance hypergeometric binomial,thm:distance binomial with different parameters} (note that we needed $\tau < 1$).
  Since the left-hand side is at least $1/3$ as shown before, we have 
  \begin{equation*}
    \frac13 \leq \frac{8 t}{N} + 2 \tau \sqrt{e},
  \end{equation*}
  so at least one of the two terms must be $1/6$ or greater.
  If $8t/N \geq 1/6$, then $t \geq N/48$ and we are done; otherwise, $\tau \geq 1/24$ and so
  \begin{equation*}
    t+2 \geq \frac{2 (N-k)}{24^2 \delta^2 k}.
  \end{equation*}
  If $k \leq N/2$, then $N-k \geq N/2$ and one deduces $t+2 \geq N/(144 \delta^2 k)$.
\end{proof}

\end{document}